\documentclass[11pt,a4paper]{article}
\usepackage{a4wide}
\usepackage{graphicx}
\usepackage{amsmath,amssymb}

\newtheorem{theorem}{Theorem}
\newtheorem{remark}{Remark}
\newtheorem{lemma}{Lemma}
\newtheorem{proposition}{Proposition}
\newtheorem{problem}{Problem}
\newtheorem{definition}{Definition}

\newenvironment{proof}
  {\par \noindent {\bf Proof\ }}
  {\hfill $\Box$ \par \medskip}

\newcommand{\cE}{\mathcal{E}}
\newcommand{\cG}{\mathcal{G}}
\newcommand{\cO}{\mathcal{O}}
\newcommand{\mC}{\mathbb{C}}
\newcommand{\mR}{\mathbb{R}}
\newcommand{\e}{\mathrm{e}}
\newcommand{\eps}{\epsilon}
\newcommand{\lam}{\lambda}

\newcommand{\weq}{w_\mathrm{eq}}

\newcommand{\bk}[1]{{\left\langle #1 \right\rangle}}
\newcommand{\abs}[1]{{\vert {#1} \vert}}
\newcommand{\norma}[1]{{\Vert {#1} \Vert}}
\newcommand{\pt}{\partial}
\newcommand{\ptt}{\frac{\pt}{\pt t}}
\newcommand{\ptx}[1]{\frac{\pt}{\pt x_{#1}}}

\DeclareMathOperator{\Exp}{\mathit{\mathcal{E}\!\mathit{xp}}}
\DeclareMathOperator{\Li}{\mathrm{Li}}
\DeclareMathOperator{\Tr}{Tr}
\DeclareMathOperator{\Op}{Op}
\DeclareMathOperator{\DIV}{div}

\numberwithin{equation}{section}
\numberwithin{theorem}{section}
\numberwithin{lemma}{section}
\numberwithin{definition}{section}
\numberwithin{proposition}{section}
\numberwithin{problem}{section}
\numberwithin{remark}{section}
\begin{document}

\begin{center}

{\bf \LARGE Derivation of isothermal quantum fluid equations 
\\[6pt] with Fermi-Dirac and Bose-Einstein statistics }

\par 
\vspace{24pt}
{\sc \large L.~Barletti}
\par 
\vspace{7pt}
Dipartimento di Matematica e Informatica ``U.Dini'', Firenze, Italy
\par 
\vspace{2pt}
{\tt \small luigi.barletti@unifi.it}

\par 
\vspace{20pt}
{\sc \large C.~Cintolesi}
\par 
\vspace{7pt}
Doctorate School in Environmental and Industrial Fluid Mechanics, \\
Universit\`a di Trieste, Italy
\par 
\vspace{2pt}
{\tt \small carlo.cintolesi@phd.units.it}

\end{center}

\vspace{8pt}

\begin{abstract}
By using the quantum maximum entropy principle we formally derive, from a underlying kinetic description, 
isothermal (hydrodynamic and diffusive) quantum fluid equations for particles with 
Fermi-Dirac and Bose-Einstein statistics. 
A semiclassical expansion of the quantum fluid equations, up to $\mathcal{O}(\hbar^2)$-terms, leads to
classical fluid equations with statistics-dependent quantum corrections, 
including a modified Bohm potential. 
The Maxwell-Boltzmann limit and the zero temperature limit are eventually discussed.  
\end{abstract}
\section{Introduction and main results}
\label{S1}
The theory of quantum fluid equations originates from a seminal paper by E.\ Madelung \cite{Madelung26}, 
who discovered that Schr\"odinger equation can be put in a hydrodynamic form 
(the {\em Madelung equations}, see Eq.\ \eqref{MadEq}).
These equations have the form of an irrotational, compressible 
and isothermal Euler system with an additional term, of order $\hbar^2$, interpreted as a ``quantum potential''
or a ``quantum pressure''.
This was later named {\em Bohm potential} after D.\ Bohm, who based on it his celebrated, although 
controversial, interpretation of quantum mechanics \cite{Bohm52a,Bohm52b,DuerrTeufel09}.
\par
Besides their undoubted theoretical importance, quantum fluid equations have become very interesting
also for applications, in particular to semiconductor devices modeling \cite{Jungel09}. 
Indeed, the fluid description of a quantum system has many practical advantages.
Not only it provides a description in terms of macroscopic variables with a direct physical interpretation
(such as density, current, temperature) but it is also amenable to semiclassical approximations, usually 
leading to fluid equations in a quasi-classical form (classical fluid equations with quantum corrections).
Such quasi-classical form is particularly suited for modeling purposes since it can be easily 
``contaminated'' with phenomenological elements (boundary conditions, external couplings and so on), 
that would be very difficult to incorporate within a purely quantum mechanical framework.
\par
Madelung equations describe the evolution of a pure (i.e.\ non-statistical) quantum state and 
(being basically equivalent to Schr\"odinger equation \cite{GasserMarkowich97})
are formally closed.
However, the most interesting case for applications is usually that of statistical systems,
in which any description in terms of a finite number of macroscopic moments is, in general,
not closed. 
Then, analogously to what happens in classical statistical mechanics, a central 
problem in the theory of quantum fluids is the closure of the moment equations. 
A commonly accepted solution to this problem is furnished by the quantum version, due to P.\ Degond
and C.\ Ringhofer \cite{DR03}, of the {\em maximum entropy principle}, a well-known paradigm 
from information theory, widely used in classical statistical mechanics and thermodynamics \cite{Levermore96}
as well as in many other disciplines (e.g.\ signal analysis).
By using the quantum maximum entropy principle (QMEP), various kind of quantum fluid models have been
deduced: drift-diffusion and energy transport \cite{DMR05}, SHE-model \cite{BourgadeEtAl06},
isothermal hydrodynamic \cite{DGM07,JM05}, non-isothermal hydrodynamic \cite{JMM06}, 
viscous hydrodynamic (Navier-Stokes) \cite{BrullMehats10}, spin (or pseudo-spin) drift-diffusion 
\cite{BM10,BF10} and hydrodynamic \cite{Zamponi12}.
Many other references can be found in Refs.\ \cite{Jungel09,Jungel12}.
\par
Although the QMEP was originally stated for a general convex entropy functional
\cite{DR03}, in all the quoted references explicit models are deduced only for Maxwell-Boltzmann statistics.%
\footnote{Reference \cite{DGM07} is a partial exception, because the fully-quantum model is deduced 
for a generic entropy. 
However, the semiclassical expansion assumes Boltzmann entropy.}
In the present paper we consider the QMEP for an entropy function
that incorporates different particle statistics (Fermi-Dirac, Maxwell-Boltzmann, Bose-Einstein, 
see Eq.\ \eqref{entropy}).
Then, we derive the corresponding isothermal fluid equations 
(drift-diffusion and hydrodynamic) and compute their explicit semiclassical expansions up to
$\cO(\hbar^2)$ terms.
\par
To our knowledge,  in the framework of QMEP, Fermi-Dirac and Bose-Einstein statistics have been so far 
considered  in Refs.\ \cite{TR10,TR11}, where a hierarchy of moment equations, in the spirit
of extended thermodynamics, is derived. 
Explicit (or partially explicit) semiclassical equations, with $\cO(\hbar^2)$ terms, 
are computed  for the first levels of the hierarchy (including the isothermal equations considered in the 
present paper).  
However, having assumed that quantum corrections only depend on the density and its derivatives, some terms
are missing that depend on the derivatives of the current (namely, the rotational terms of order $\hbar^2$ 
that appear instead in our Eqs.\ \eqref{semi_hydro}, \eqref{MB_hydro} and \eqref{T0_FD_hydro}).
Also Ref.\  \cite{JKP11} is worth to be mentioned, where a hierarchy of diffusive moment equations with 
Fermi-Dirac statistics is derived in the semiclassical limit (no $\cO(\hbar^2)$ corrections).
\par
Our derivation starts from a quantum kinetic level, represented by a one-particle, $d$-dimensional, 
Wigner equation \cite{Folland89,Wigner32,ZachosEtAl05}, with a BGK collisional term 
(see Eq.\ \eqref{WBGK}) that relaxes the system to a local equilibrium.
The local equilibrium Wigner function $\weq$ is assumed to be given by the QMEP which, generally speaking,
stipulates that the local equilibrium maximizes an entropy functional under the constraint that some of 
its macroscopic moments are given. 
Which moments are constrained depends on which kind of fluid equations we are interested in.
In our case, the moments are the particle density $n$, for the diffusive equations, and, in addition, the current 
$J= (J_1,\ldots, J_d)$  (or, equivalently, the velocity $u = J/n$) for the isothermal hydrodynamic equations.  
We choose an entropy functional (see Eq.\ \eqref{Edef}) that contains the information on the particle statistics 
and also accounts for the fixed equilibrium temperature (and, therefore, is a free-energy rather than
an entropy).
As we shall see in Subsec.\ \ref{Sec_MEP}, a solution $\weq$ of the constrained minimization problem 
can be formally written and depends on $d+1$ Lagrange multipliers, $A$ and $B = (B_1,\ldots,B_d)$, which
are implicitly related to the moments $n$ and $J$ because of the constraints.
After a suitable scaling of the Wigner-BGK equation, we identify two dimensionless parameters: a scaled relaxation 
time $\alpha$ and a scaled Planck constant $\eps$.
Then, in the limit of vanishing $\alpha$ (i.e., basically, assuming that the system has relaxed to the state $\weq$)
we can deduce moment equations for $n$ (diffusive) and for $n$ and $J$ (hydrodynamic), where the
extra moments are expressed in function of $A$ and $B$.
Therefore, the moments equations are formally closed because of the constraint relations.
The fully-quantum models obtained in this way, given by Eqs.\ \eqref{QHD} and \eqref{QDD}, are very implicit and, 
therefore, it is reasonable to look for approximated, but explicit, equations.
In particular, we look for the semiclassical approximation of Eqs.\ \eqref{QHD} and \eqref{QDD},
by assuming $\eps$ small. 
\par
The semiclassical approximation of the quantum fluid models requires the semiclassical expansions 
of $\weq$, which is particularly natural in the Wigner-Weyl-Moyal formalism \cite{Folland89,ZachosEtAl05}.
This expansion, which involves an interesting application of the Moyal calculus, as well as the computation of 
a variety of integrals of Fermi and Bose type (see Appendix \ref{polylog}), is carried out explicitly up
to order $\cO(\eps^2)$ and leads to the main result of the paper, represented by Theorems \ref{main_thm_1} 
and \ref{main_thm_2}.
The semiclassical equations that we obtain, Eqs.\ \eqref{semi_hydro} and \eqref{semi_diffu}, 
are the generalization of the semiclassical diffusive and hydrodynamic equations
derived in Refs.\ \cite{DGM07,DMR05,JM05} for Maxwell-Boltzmann statistics.
They have the form of their classical counterparts (compressible isothermal Euler system and drift-diffusion
equations) with quantum corrections (terms of order $\cO(\eps^2)$) as well as corrections coming from 
the particle statistics.
In particular, a term corresponding to a modified Bohm potential can be identified (see Eq.\ \eqref{Qdef}). 
It is worth remarking that in the Bose-Einstein case, and dimension $d \geq 3$, 
our equations are only valid by assuming that the fluid is entirely in the non-condensate phase or, equivalently, 
that the temperature  is uniformly supercritical (see Proposition  \ref{mu0} and  Remark \ref{rem_condens}).
\par
In the last part of the paper we perform a formal analysis of specific physical regimes, where
Eqs.\ \eqref{semi_hydro} and \eqref{semi_diffu} take particular forms.
We discuss the case of irrotational fluids, the Maxwell-Boltzmann limit (recovering the equations
derived in Refs.\ \cite{DGM07,DMR05,JM05}), and the vanishing-temperature limit.
The latter is particularly interesting in the Fermi-Dirac case (the only one in which 
Eqs.\  \eqref{semi_hydro} and \eqref{semi_diffu} have a regular behavior as $T \to 0$) and leads to 
Eqs.\ \eqref{T0_FD_hydro} and \eqref{T0_FD_diffu}, describing a so-called 
``completely degenerate fluid''. 
These equations contain power-law diffusive and rotational terms,
and a limit Bohm potential which differs from the usual one by just a
dimension-dependent constant factor.
The $T\to 0$ limit in the Maxwell-Boltzmann and Bose-Einstein cases shows a singular behavior and
reasonable results can only be given for the Bose-Einstein case with $d \leq 2$ and assuming 
that the fluid is irrotational.
\par
\smallskip
The outline of the paper is as follows.
Section \ref{S2} is devoted to the derivation of the quantum fluid equations:
in Subsecs.\ \ref{Sec_WBGK} and  \ref{Sec_Scaling} we introduce the kinetic Wigner-BGK equation and
its hydrodynamic and diffusive scalings;
in Subsec.\ \ref{SS2.3} we derive the equations for the moments $n$ and $J$ and, in Subsec.\ \ref{Sec_MEP}, 
we perform their formal closure by using the QMEP.
Section \ref{S3} is devoted to the semiclassical approximation of the quantum fluid models derived in the preceding section:
in Subsec.\ \ref{S3.1} we compute the semiclassical expansion of the local equilibrium Wigner function and, in particular,
of the Lagrange multipliers $A$ and $B$ as functions of the moments $n$ and $J$;
in Subsec.\ \ref{Sec_SFE} the expansion of $A$ and $B$ is used to derive the semiclassical equations 
\eqref{semi_hydro} and \eqref{semi_diffu}.
Finally, Sec.\ \ref{S4} is devoted to the analysis of the above-mentioned particular regimes: 
the irrotational fluid (Subsec.\ \ref{S4.1}), the Maxwell-Boltzmann limit (Subsec.\ \ref{S4.2}) and the 
zero-temperature limit (Subsec.\ \ref{S4.3}).
Some technical material has been placed in two appendices:
Appendix \ref{polylog} contains generalities about Fermi and Bose integrals as well as the computation
of related integrals that have been encountered in the paper;
Appendix \ref{AppB} contains some postponed proofs. 
\section{Quantum fluid equations}
\label{S2}
In this section we derive fully quantum fluid-dynamic equations of two types: isothermal hydrodynamic equations
and diffusive equations.
Such derivations, and the structure of the resulting equations, do not differ significantly from what is already well known in literature \cite{DGM07,Jungel09,Jungel12}.
The specific role played by statistics (and, therefore, the novelty of the present work) will be more explicit in the semiclassical 
expansion of the equations, which will be carried out in the remainder of the paper.
\subsection{The Wigner-BGK equation}
\label{Sec_WBGK}
The starting point of our derivation is the kinetic description of a one-particle quantum statistical state, 
given in terms of one-particle Wigner functions \cite{Wigner32,ZachosEtAl05}. 
Let us now briefly recall the basic definitions and properties.
\par
A mixed (statistical) one-particle quantum state for an ensemble of scalar particles  in $\mR^d$ 
(where  $d$ = 1,2 or 3 are the interesting values, e.g.\ for nano-electronic applications),
is described by a density operator $\varrho$, i.e.\ a bounded non-negative operator with unit trace, acting on $L^2(\mR^d,\mC)$.
The associated Wigner function, $w = w(x,p)$, $(x,p) \in \mR^{2d}$,  is given by the inverse Weyl quantization of $\varrho$,
\begin{equation}
\label{Wdef}
  w = \Op_\hbar^{-1}(\varrho),
\end{equation}
where the Weyl quantization of a phase-space function  (a ``symbol'') $a = a(x,p)$ is the operator $\Op_\hbar(a)$ 
formally defined by
\begin{equation}
\label{Weyl}
  \left[ \Op_\hbar(a)\psi \right](x)  =  \frac{1}{(2\pi\hbar)^d} \int_{\mR^{2d}} 
  a\left( \frac{x+y}{2}, p \right)\,\psi(y)\,\e^{i(x-y)\cdot p/\hbar}\,dy\,dp.
\end{equation}
The Wigner function has a more direct definition as the ``Wigner transform''  
\begin{equation}
\label{WigTransf}
  w(x,p) = \int_{\mR^d} 
  \varrho\left( x + \frac{\xi}{2},  x - \frac{\xi}{2} \right) \e^{-ip \cdot \xi/\hbar}  d\xi,
\end{equation}
of the density matrix $\varrho(x,y)$,
i.e.\ (with a little abuse of notation) the integral kernel of the density operator $\varrho$.
\par
The operator product translated at the level of symbols leads to the definition of Moyal (or ``twisted'') product
\begin{equation}
\label{Moyal}
   a \# b = \Op_\hbar^{-1} \left( \Op_\hbar(a) \Op_\hbar(b) \right),
\end{equation}
which possesses the formal semiclassical expansion 
\begin{equation}
\label{MoyalExpansion}
\begin{aligned}
 &a \# b = \sum_{k=0}^\infty \hbar^k a \#_k b, 
 \\
 &a \#_k b =  \frac{1}{(2i)^k} \sum_{\abs{\alpha} + \abs{\beta} = k}
\frac{(-1)^{\abs{\alpha}} }{ \alpha!\, \beta! }
\left(\nabla_x^\alpha \nabla_p^\beta a \right)\left(\nabla_p^\alpha \nabla_x^\beta b \right).
\end{aligned}
\end{equation}
In particular, $\#_0$ is the usual function product, $a \#_0 b =  ab$,
and $\#_1$ is the Poisson bracket 
\begin{equation}
\label{Moyal_1}
  a \#_1 b =   \frac{i}{2} \sum_{k=1}^{d}
   \left( \frac{\pt a}{\pt x_k} \frac{\pt b}{\pt p_k} - \frac{\pt a}{\pt p_k}\frac{\pt b}{\pt x_k} \right).
\end{equation}
The dynamics of the time-dependent Wigner function $w(t) = w(x,p,t)$ can be immediately deduced from the
dynamics of the corresponding density operator $\varrho(t)$, i.e.\ from the von Neumann equation
 (Schr\"odinger equation for mixed states)
\begin{equation}
\label{VNE}
  i\hbar\,\ptt \varrho(t) = \left[H,\varrho(t)\right] := H\varrho(t) - \varrho(t)H,
\end{equation}
where $H$ denotes the Hamiltonian operator of the system.
If $h = \Op_\hbar^{-1}(H)$ is the symbol of $H$, then, from Eqs.\ \eqref{VNE} and \eqref{Moyal} we obtain
the ``Wigner equation''
\begin{equation}
\label{WE1}
  i\hbar\,\ptt w(t) = \left\{h,w(t)\right\}_\# := h\# w(t) - w(t)\#h.
\end{equation}
Taking $h$ as the standard hamiltonian symbol 
\begin{equation}
\label{hDef}
 h(x,p) = \frac{\abs{p}^2}{2m} + V(x)
 \end{equation}
(where $m$ is the particle effective mass and $V$ is a one-particle potential), the Wigner equation \eqref{WE1} 
can be written in the more explicit form 
\begin{equation}
\label{WE}
  \ptt w(t) + \frac{p}{m}\cdot\nabla_x w(t) + \Theta_\hbar[V] w(t) = 0,
\end{equation}
where $\Theta_\hbar[V]w(t) = \frac{i}{\hbar}  \left\{V,w(t) \right\}_\#$ is given by
\begin{multline*}
\left[\Theta_\hbar[V]w(t)\right](x,p)  = 
\\
\frac{i}{\hbar} \int_{\mR^{2d}}\left[V\left(x+\frac{\xi}{2}\right)-V\left(x-\frac{\xi}{2}\right) \right]
\e^{i\xi\cdot(p'-p)/\hbar}\,w(x,p',t)\, \frac{d\xi\,dp'}{(2\pi\hbar)^d}.
\end{multline*}
\par
One of the most interesting properties of the Wigner function is that its moments
have a direct physical interpretation in terms of macroscopic fluid quantities, which makes Wigner 
functions an ideal tool for the derivation of quantum fluid equations.
In this paper we shall write equations, in different fluid regimes, for the first $1+d$ moments: the density
\begin{equation}
\label{nDef}
  n(x,t) = \frac{1}{(2\pi\hbar)^{d}}\int_{\mR^d}  w(x,p,t)\,dp = \varrho(x,x,t)
\end{equation}
and the $d$ components of the current
\begin{equation}
\label{JDef}
 J_k(x,t) = \frac{1}{(2\pi\hbar)^{d}}\int_{\mR^d}  p_k w(x,p,t)\,dp
 = \frac{\hbar}{2i} \left(\frac{\pt \varrho}{\pt x_k} - \frac{\pt \varrho}{\pt y_k} \right)(x,x,t),
\end{equation}
(where the corresponding expressions in terms of the time-dependent density matrix 
$\varrho(x,y,t)$ have also been shown).
\begin{remark}
\label{rem_phys_wig}
Our choice of defining the Wigner function as the inverse Weyl quantization of the density operator
implies that $w$ is a dimensionless quantity (this is apparent from Eq.\ \eqref{WigTransf}, recalling that
the density matrix $\varrho(x,y)$ has the physical dimensions of a number density in position space). 
However, the usual ``physical'' definition of Wigner function \cite{Wigner32,ZachosEtAl05} 
requires an extra factor $1/(2\pi\hbar)^d$,
so that the physical Wigner function has the dimensions of a number density in phase space.
This is the reason of the factor $1/(2\pi\hbar)^d$ appearing in Eqs.\ \eqref{nDef} and \eqref{JDef}.
\end{remark}
Now, by definition, a system is driven to a fluid regime by collisions. 
Following Degond, Ringhofer and M\'ehats \cite{DMR05,DR03}, we endow the Wigner equation \eqref{WE} 
with a collision mechanism of BGK type \cite{Arnold96}
\begin{equation}
\label{WBGK}
  \frac{\pt w}{\pt t} + \frac{p}{m}\cdot\nabla_x w + \Theta_\hbar[V] w =\frac{1}{\tau}\left(\weq[w] - w\right).
\end{equation}
Here, $\tau$ is a typical relaxation time and $\weq[w]$ is a Wigner function which represent the local equilibrium state
reached by the system because of collisions. 
As we shall see in the following, the central point of the whole derivation is that $\weq[w]$ is assumed to be the 
maximizer of a suitable quantum entropy functional, subject to the constraint of sharing certain moments with $w$ 
(namely, $n$ and $J$ in the isothermal hydrodynamic case, and $n$ in  the diffusive case). 
We remark that the one-particle potential $V$ accounts for other kinds of interactions, including mean-field Poisson or Hartree-like \cite{TR10,TR11} potentials.
\subsection{Scaling the Wigner-BGK equation}
\label{Sec_Scaling}
In order to write Eq.\ \eqref{WBGK} in the hydrodynamic and diffusive scalings, let us introduce a reference 
length $x_0$, time $t_0$ and energy $E_0$.
Reference temperature and momentum are naturally related to $E_0$ by
$$
  k_B T_0 = E_0, \qquad \frac{p_0^2}{m} =  E_0,
$$
where $k_B$ is the Boltzmann constant.
Then, in Eq.\ \eqref{WBGK} we  switch to dimensionless quantities
$$
  x \to x_0x, \qquad t \to t_0t, \qquad p \to p_0p, \qquad V \to E_0V, 
$$
(for the sake of simplicity the new dimensionless variables are denoted by the same symbols as 
the old ones), which yields
$$
  \frac{1}{t_0}\,\frac{\pt w}{\pt t} + \frac{p_0}{mx_0}p \cdot \nabla_x + \frac{E_0}{x_0p_0} \, \Theta_{\frac{\hbar}{x_0p_0}}[V] 
  = \frac{1}{\tau}\left(\weq[w] - w\right).
$$
Note that we have not to rescale the Wigner functions $w$ and $w_{\weq}$ because we are already using
dimensionless Wigner functions (see Remark \ref{rem_phys_wig}). 
We rewrite the last equation by introducing the {\em semiclassical parameter}
\begin{equation}
\label{epsDef}
 \eps = \frac{\hbar}{x_0 p_0}
\end{equation}
and the energy time scale 
$$
  t_E = \frac{mx_0}{p_0}
 $$ 
 (i.e.\ the order of time for a particle of kinetic energy $E_0$ to travel a distance $x_0$), obtaining:
\begin{equation}
\label{scaling_general}
  \frac{1}{t_0}\,\frac{\pt w}{\pt t} + \frac{1}{t_E}\,p \cdot \nabla_x w +   \frac{1}{t_E}\, \Theta_{\eps}[V] w = 
   \frac{1}{\tau} \left(\weq[w] - w\right).
\end{equation}
Now, two different scaling assumptions, corresponding to different fluid regimes, 
can be made.
\paragraph{Hydrodynamic regime.}
In this regime the system is observed on the time-scale $t_E$ 
and collisions are assumed to act on a much shorter time-scale; 
then we put
\begin{equation}
\label{scaling_hydro}
 \alpha := \frac{\tau}{t_E} \ll 1, \qquad  t_0 = t_E.
\end{equation}
The corresponding Wigner-BGK equation takes therefore the {\em hydrodynamic scaling} form:
\begin{equation}
\label{WBGK_hydro}
 \alpha\,\frac{\pt w}{\pt t} + \alpha\,p \cdot \nabla_x w +  \alpha\, \Theta_{\eps}[V] w = \weq[w] - w.
\end{equation}
\paragraph{Diffusive regime.}
In this regime the collisions are still assumed to act on a time-scale much shorter than $t_E$,  but
the system is observed on a time-scale much larger than $t_E$; then we put
\begin{equation}
\label{scaling_diffu}
 \alpha := \frac{\tau}{t_E} \ll 1,  \qquad \frac{t_E}{t_0} = \alpha
\end{equation}
(so that $t_0 = t_E^2/\tau$). 
The corresponding Wigner-BGK equation takes in this case the {\em diffusive scaling} form:
\begin{equation}
\label{WBGK_diffu}
 \alpha^2\,\frac{\pt w}{\pt t} + \alpha\,p \cdot \nabla_x w +  \alpha\, \Theta_{\eps}[V] w = \weq[w] - w.
\end{equation}
\par
Note, in both cases, the presence of two dimensionless parameters: $\eps$ and $\alpha$.
In the remainder of this section we shall deal with the fluid asymptotics, $\alpha \to 0$, 
leaving $\eps$ untouched; then, in the following sections, we shall work on the semiclassical 
expansion of the fluid equations for small $\eps$.
\begin{remark}
\label{ScaledWeyl}
In the new dimensionless variables, all the identities involving Weyl quantization and Moyal product
are obtained from the original ones by the formal substitution $\hbar \mapsto \eps$.
\end{remark}
\subsection{Derivation of quantum fluid equations}
\label{SS2.3}
First of all, let us introduce the short notation
$$
   \bk{w}(x,t) :=  \int_{\mR^d} w(x,p,t)\,dp.
$$
If $w$ is the Wigner function of our particle system, we are going to write down equations for the 
moments $n = \bk{w}$ and $J = \bk{p w}$, but we have to keep in mind that
the {\em true} density and current are $N_0 n$ and $p_0 N_0  J$, where 
\begin{equation}
\label{N0def}
  N_0 = \left(  \frac{p_0}{2\pi\hbar} \right)^d  = \left(  \frac{m k_B T_0}{(2\pi\hbar)^2}  \right)^{d/2}
\end{equation}
(see Eqs.\ \eqref{nDef} and \eqref{JDef}, and Remark \ref{rem_phys_wig}).
\subsubsection{Isothermal quantum hydrodynamic equations}
In order to derive isothermal hydrodynamic equations from Eq.\ \eqref{WBGK_hydro}, we have to 
assume that collisions conserve the number of particles and their momentum while making the system
relax towards a local equilibrium state $\weq[w]$ (to be completely described later on) 
at a constant temperature $T_\mathrm{ext}$.
Thus, we have to impose on $\weq[w]$ the moment constraints
\begin{equation}
\label{constraints_hydro}
 \bk{\weq[w]}  = n = \bk{w},
 \qquad
  \bk{p_i \weq[w]}  = J_i = \bk{p_i w},
\end{equation}
for  $i = 1,\ldots,d$.
Let now $w_\alpha$ be solution of Eq.\ \eqref{WBGK_hydro} and assume that the limit $w_\alpha \to w_0$ 
for $\alpha \to 0$ exists with finite moments $n = \bk{w_0}$ and $J = \bk{p w_0}$. 
Then, from \eqref{WBGK_hydro} and \eqref{constraints_hydro} we get $w_0 = \weq[w_0]$.
Taking the moments of both sides of Eq.\ \eqref{WBGK_hydro} and letting $\alpha \to 0$ we obtain 
$$
\begin{aligned}
 &\ptt \bk{\weq[w_0]} +  \ptx{i}  \bk{p_i \weq[w_0]} +  \bk{ \Theta_{\eps}[V] \weq[w_0]} = 0,
 \\
 &\ptt \bk{p_i \weq[w_0]} +  \ptx{j}  \bk{p_i p_j \weq[w_0]} +  \bk{ p_i \Theta_{\eps}[V] \weq[w_0]} = 0,
\end{aligned}
 $$
where the summation convention on repeated indices has been assumed.
From the semiclassical expansion of the potential operator,
\begin{equation}
\begin{aligned}
  \Theta_\eps[V] &= \frac{i}{\eps}  \left\{V,w(t) \right\}_\# 
   =  \frac{i}{\eps} \left( V \# w - w\# V\right)
\\
  &= - \sum_{k=0}^\infty (-1)^k \left(\frac{\eps}{2}\right)^{2k}  \sum_{\abs{\alpha}=2k+1} 
    \nabla_x^\alpha V \, \nabla_p^\alpha w
\end{aligned}
\end{equation}
(where \eqref{MoyalExpansion} was used, see also Remark \ref{ScaledWeyl}), we immediately obtain
\begin{equation}
\label{momTheta}
     \bk{ p_i \Theta_{\eps}[V] \weq[w_0]} =   \bk{ \weq[w_0]} \frac{\pt V}{\pt x_i} =  n\,\frac{\pt V}{\pt x_i},
\end{equation}
and, then, the moment equations read as follows:
\begin{equation}
\label{open_hydro_eqs}
\begin{aligned}
 &\frac{\pt n}{\pt t} +  \frac{\pt J_i}{\pt x_i} = 0
 \\
 &\frac{\pt J_i}{\pt t} +  \ptx{j}  \bk{p_i p_j \weq[w_0]} + n\, \ptx{i} V  = 0.
\end{aligned}
\end{equation}
If $\weq[w_0]$ can be uniquely specified as  a function of $n$ and $J$, from the constraints \eqref{constraints_hydro},
then the system \eqref{open_hydro_eqs} is formally closed. 
\subsubsection{Quantum diffusive equations}
Diffusive equations can be obtained from Eq.\ \eqref{WBGK_diffu} by using the ``Chapman-Enskog'' method.
We now only assume that collisions conserve the number of particles, which leads to the unique constraint
\begin{equation}
\label{constraints_diffu}
 \bk{\weq[w]}  = n = \bk{w}.
\end{equation}
Then, we assume that the solution $w_\alpha$ of Eq.\ \eqref{WBGK_diffu}, for $\alpha \to 0$, 
 has a limit  $w_\alpha \to w_0$ with finite density $n = \bk{w_0}$.
Letting $\alpha \to 0$ in Eq.\ \eqref{WBGK_diffu} we still obtain $w_0 = \weq[w_0]$ but, contrarily to the previous case, 
the equation for the density 
\begin{equation}
\label{diffu_aux}
 \alpha \,\ptt \bk{w_\alpha} +  \ptx{i}  \bk{p_i w_\alpha}  = 0
\end{equation}
only gives, in the limit, the condition
\begin{equation}
\label{CEconstraint}
 \bk{p \weq[w_0]}  = 0,
\end{equation}
i.e.\ the equilibrium state carries no current.
The diffusive equations must be sought at next order of the Chapman-Enskog expansion
$$
   w_\alpha = \weq[w_\alpha] + \alpha w_1.
$$
Substituting this ansatz into Eq.\ \eqref{WBGK_diffu}, and letting $\alpha \to 0$, yields
$$
  w_1 = - \left( p \cdot \nabla_x +  \Theta_{\eps}[V] \right)\weq[w_0] 
$$
and, therefore, from Eqs.\ \eqref{diffu_aux} and \eqref{momTheta}, we obtain the diffusive equation
\begin{equation}
\label{open_diffu_eqs}
\frac{\pt n}{\pt t} =  \ptx{i} \left( J_i  + n\, \frac{\pt V}{\pt x_i} \right),
 \qquad
 J_i =  \ptx{j}  \bk{p_i p_j \weq[w_0]} .
\end{equation}
Once again, if $\weq[w_0]$ can be uniquely specified as  a function of $n$ from the constraint \eqref{constraints_diffu},
then Eq.\ \eqref{open_diffu_eqs} is formally closed.
\par
The ``closure'' of Eqs.\ \eqref{open_hydro_eqs} and \eqref{open_diffu_eqs} will be the central issue of the remainder
of the paper.
\subsection{Maximum entropy closure}
\label{Sec_MEP}
Following Refs.\ \cite{DMR05,DR03}, we assume that the local equilibrium state $\weq[w]$ satisfies 
a quantum maximum entropy principle (QMEP), which basically states that 
$\weq[w]$ is the most probable state compatible with the information we have about it.
In our case, such information is: 
\begin{enumerate}
\item
the temperature has a constant value $T_\mathrm{ext}$;
\item
collisions conserve the number of particles;
\item
collisions conserve also the momentum in the hydrodynamic regime.
\end{enumerate}
Point 1 implies that $\weq[w]$ should be a minimizer of the free-energy  
(rather than a maximizer of the entropy) \cite{DMR05}.
Points 2 and 3 imply that $\weq[w]$ is subject to the constraints \eqref{constraints_hydro} in the hydrodynamic
case, or to the single constraint \eqref{constraints_diffu} in the diffusive case.
\par
Let  $s : \mR^+ \to \mR$ be a regular, convex, function and $\varrho$ a density operator.
We can define the {\em von Neumann entropy} \cite{vonNeumann} of the state $\varrho$ 
as $\Tr\{ -s(\varrho) \}$, where $s(\varrho)$ is given by the functional calculus on self-adjoint 
operators and $\Tr$ denotes the operator trace.
It is worth remarking that we are using dimensionless quantities (see Subsection \ref{Sec_Scaling}); 
the dimensional definition of entropy would be $\Tr\{ -k_B s(\varrho/N_0) \}$, where $N_0$ is given by \eqref{N0def}.
The corresponding free-energy at temperature $T_\mathrm{ext}$ is
\begin{equation}
\label{Edef}
 \cE(\varrho) = \Tr\left\{ H\varrho + T s(\varrho)\right\},
\end{equation}
where $H$ is the (scaled) Hamiltonian and 
\begin{equation}
\label{Tdef}
  T = \frac{T_\mathrm{ext}}{T_0} =  \frac{k_B T_\mathrm{ext}}{E_0}
\end{equation}
is the scaled external temperature.
The Wigner-Weyl correspondence (see Subsection \ref{Sec_WBGK}) allows us to define the
free-energy of a Wigner function $w$ simply as $\cE\left( \Op_\eps(w) \right)$.
\par
All this considered, we shall assume the local-equilibrium Wigner function $\weq[w]$ to be 
solution of the following constrained  minimization problem.
\begin{problem}
\label{MEP}
Let $n = \bk{w}$ and $J = \bk{p w}$. 
Find a Wigner function $\weq[w]$ that minimizes the functional 
$\cE\left( \Op_\eps(f) \right)$ among all Wigner functions $f$ 
that satisfy 
$$
\begin{aligned}
 & \bk{f} = n, \quad  \bk{p f} = J   \qquad &\text{(hydrodynamic case)}
 \\
 &\bk{f} = n, &\text{(diffusive case)}.
\end{aligned}
$$
\end{problem}
In Ref.\ \cite{DR03} is formally proven the following necessary condition%
\footnote{
A rigorous proof of existence and uniqueness of the constrained minimization problem has been
recently obtained by M\'ehats and Pinaud \cite{MP10} for the one-dimensional case with periodic boundary conditions.
}
for $\weq[w]$.
\begin{theorem}
\label{T1}
A necessary condition for $\weq[w]$ to be a solution of Problem  \ref{MEP}
is that $d+1$  Lagrange multipliers $A$ and $B = (B_1,\ldots,B_d)$, functions of $x$ and $t$,
exist such that
$$
   \weq[w] = \cG_{A,B},
$$
where
\begin{equation}
\label{ggeneric}
 \cG_{A,B} =  \Op_\eps^{-1} \left\{ (s')^{-1} \left( \frac{\Op_\eps\left(h_{A,B} \right)}{T}  \right) \right\},
\end{equation} 
and
\begin{equation}
\label{hABdef}
\begin{aligned}
 &h_{A,B}(x,p,t) = \frac{\abs{p-B(x,t)}^2}{2} - A(x,t), \quad &\text{(hydrodynamic case)}
 \\
 &h_{A,B}(x,p,t) = \frac{\abs{p}^2}{2} - A(x,t), &\text{(diffusive case)}.
\end{aligned}
\end{equation}   
\end{theorem}
Of course, the Weyl quantization $\Op_\eps$ acts on functions of $x$ and $p$ (see Subsec.\ \ref{Sec_WBGK}),
$t$ being just a parameter.
Note that the Lagrange multipliers $A, B_1, \ldots, B_d$ furnish the necessary degrees of freedom to satisfy the constraints.
Note also that the hydrodynamic case contains the diffusive as a particular case corresponding to $B=0$.
This fact allows us to treat the two cases at once, the latter being simply obtained by taking $B=0$.
The Lagrange multiplier $A$ is the so-called chemical potential.
\par
So far, $s$ is a generic entropy function (minus the entropy, to be precise).
A further piece of information, namely the statistics of indistinguishable particles, can be inserted by choosing 
a suitable form of $s$.
In this paper we consider a typical family of entropy functions, dependent on the real parameter $\lam$, 
of the form
\begin{equation}
\label{entropy}
 s(f) =  f\log f + \lam^{-1} (1 - \lam f)\log(1- \lam f)
\end{equation}   
(where, for $\lam = 0$,  $s(f) = f \log f - f$ has to be intended as a limit).
For such $s$ we have
\begin{equation}
 (s')^{-1}(h) = \frac{1}{\e^{h} + \lam}
\end{equation}   
and, then, Eq.\ \eqref{ggeneric} is specialized in
\begin{equation}
\label{MEPg}
 \cG_{A,B} =  \Op_\eps^{-1} \left\{  \left[   \exp \left(\frac{\Op_\eps\left(h_{A,B} \right)}{T} \right) + \lam\right]^{-1} \right\}.
\end{equation} 
As usual, the parameter $\lam$ has been introduced in order to consider different cases at once, 
the most important being of course:
$$
  \lam = 
  \left\{
\begin{aligned}
 1,& \quad &\text{Fermi-Dirac (FD) statistics,}
\\ 
  0,& \quad &\text{Maxwell-Boltzmann (MB) statistics,}
\\
  -1,& \quad &\text{Bose-Einstein (BE) statistics.}
\end{aligned}
\right.
$$   
\par
The quantum hydrodynamic/diffusive equations \eqref{open_hydro_eqs}/\eqref{open_diffu_eqs} 
can now be formally closed by assuming that the local equilibrium Wigner function $\weq[w]$ 
is given by the QMEP.
Then, according to Theorem \ref{T1}, $\weq[w] = \cG_{A,B}$,
where the  Lagrange multipliers $A$ and $B$ are related to the moments $n = \bk{w}$ and $J = \bk{pw}$
through the constraints \eqref{constraints_hydro}/\eqref{constraints_diffu}. 
Then, in Eqs.\  \eqref{open_hydro_eqs} and \eqref{open_diffu_eqs}  the extra moment 
$\bk{p_i p_j \weq[w_0]} = \bk{p_i p_j \cG_{A,B}}$
can be viewed (at least in principle) as a function of  $n = \bk{w_0}$ and $J = \bk{pw_0}$, which means
that the equations are closed.
\par
The following Proposition is proven in Ref.\ \cite{DGM07} for MB entropy, using the density-operator
formalism. 
The proof given there is indeed independent on the choice of the entropy function and so the result is certainly
valid also in the present case.
However, we decided to give a proof (in Appendix), just because it may be interesting to see how it looks
like  in the Wigner formalism.\footnote{%
The proof given in Ref.\ \cite{DGM07}, however, is still more general because the 
density-operator formalism covers the case of a system confined in a domain $\Omega \in \mR^d$, while the
Wigner formalism is valid only in the whole-space case.
}
\begin{proposition}
\label{pr_Qclosure}
Let $\cG_{A,B}$ be given by  Eq.\ \eqref{MEPg} with $\bk{\cG_{A,B}} = n$ and $\bk{p\,\cG_{A,B}} = J$;
then:
\begin{equation}
\label{Qclosure}
   \ptx{j} \bk{p_i p_j \cG_{A,B}} =
   \ptx{j} (J_i  B_j) +  \left(J_j - nB_j \right) \frac{\pt B_j}{\pt x_i} + n \frac{\pt A}{\pt x_i}.
\end{equation}
\end{proposition}
\begin{proof}
See Appendix  \ref{proof1}.
\end{proof}
Equation \eqref{Qclosure} represent the formal closure of Eqs.\ \eqref{open_hydro_eqs} 
and \eqref{open_diffu_eqs}, obtained by taking $\weq[w_0] = \cG_{A,B}$.
Then, we can conclude this section by summarizing the fully-quantum hydrodynamic
and diffusive models with FD, MB or BE statistics.\footnote{%
In the Bose-Einstein case, these models are only suited to describe the non-condensate phase;
we shall discuss this point later on, in the semiclassical framework (see Proposition \ref{mu0}
and Remark \ref{rem_condens}).
}
\paragraph{Isothermal quantum hydrodynamic model}
\begin{subequations}
\label{QHD}
\begin{align}
\label{QHDa}
&
\left\{ 
\begin{aligned}
 &\frac{\pt n}{\pt t}  +  \frac{\pt J_i}{\pt x_i}  = 0
 \\
 &\frac{\pt J_i}{\pt t}  +   \ptx{j} (J_i  B_j) +  \left(J_j - nB_j \right) \frac{\pt B_j}{\pt x_i}  + n \ptx{i} (A + V) = 0.
\end{aligned}
\right.
\\[6pt]
\label{QHDb}
 &\left\{
 \begin{aligned}
 &\bk{\cG_{A,B}} = n,
 \\  
 &\bk{p\,\cG_{A,B}} = J,
 \end{aligned}
 \right.
\\[6pt]
\label{QHDc}
 &\Op_\eps (\cG_{A,B}) =  \left[   \exp \left(   \Op_\eps \left( \frac{\abs{p-B}^2}{2T} - \frac{A}{T} \right) \right) + \lam\right]^{-1}. 
\end{align}
\end{subequations}
\paragraph{Quantum diffusive model}
\begin{subequations}
\label{QDD}
\begin{align}
\label{QDDa}
&\frac{\pt n}{\pt t} =  \ptx{i} \left(  n \ptx{i} (A+V) \right),
\\[6pt]
\label{QDDb}
 &\bk{g_A} = n, 
\\[6pt]
\label{QDDc}
 &\Op_\eps (g_A) =  \left[   \exp \left(   \Op_\eps \left( \frac{\abs{p}^2}{2T} - \frac{A}{T} \right) \right) + \lam\right]^{-1}.
\end{align}
\end{subequations}
Note that the two models are constituted by the moment equations and an implicit relation
linking the Lagrange multipliers to the moments.
Such structure is rather involved, because it implies solving Eq.\ \eqref{QHDb} or Eq.\ \eqref{QDDb} for $A$ and $B$,
where $\cG_{A,B}$ is a complicated object defined in terms of back and forth Weyl quantization.
Although models of this kind are amenable to numerical implementation \cite{Gallego04,Gallego05}, it is certainly 
convenient to look for approximated, but more explicit, models. 
In particular,  it is natural to look for a semiclassical ($\eps \ll 1$) approximation of systems 
\eqref{QHD} and \eqref{QDD}.
This is what the remainder of the paper will be devoted to.
\section{Semiclassical approximation of the fluid models}
\label{S3}
In this section we perform a formal semiclassical expansion of the quantum hydrodynamic and diffusive 
models \eqref{QHD} and  \eqref{QDD} up to order $\eps^2$.
We shall derive explicit expressions of the Lagrange multipliers $A$ and $B = (B_1,\ldots, B_d)$ as
functions of the unknown moments $n$ and $J= (J_1,\ldots, J_d)$ neglecting terms of order higher than $\eps^2$ 
(as we shall see, the neglected terms will be actually of order $\eps^4$).
Once such expressions are obtained, they can be substituted for $A$ and $B$ 
in Eqs.\  \eqref{QHDa}/\eqref{QDDa}, yielding therefore semiclassical hydrodynamic/diffusive equations. 
\subsection{Semiclassical expansion of the Lagrange multipliers}
\label{S3.1}
The first step is the computation of the semiclassical expansion of 
$\cG_{A,B}$, the local equilibrium Wigner function given by Eqs.\ \eqref{MEPg} and \eqref{hABdef}.
We recall that the hydrodynamic and the diffusive cases can be unified at this stage, since the latter 
corresponds to the special case $B=0$.
Then, throughout this subsection, we shall work with the most general equilibrium function $\cG_{A,B}$.
\par
\smallskip
In order to avoid cumbersome notations, let us simply denote $\cG_{A,B}$ by $\cG$ and define
\begin{equation}
\label{newh}
 h := \frac{h_{A,B}}{T} = \frac{\abs{p-B}^2}{2T} - \frac{A}{T}
\end{equation}
(not to be confused with the symbol $h$ used for original Hamiltonian \eqref{hDef}).
The following proposition, based on a simple remark, is actually the key point for the computation 
of the semiclassical expansion.
\begin{proposition}
\label{keyprop}
Let us temporarily assume that $A$ and $B$ do not depend on $\eps$ and
consider the formal semiclassical expansion of $\cG \equiv \cG_{A,B}$:
\begin{equation}
\label{semiG}
  \cG = \cG_0 + \eps \cG_1 + \eps^2 \cG_2 + \cdots.
\end{equation}
Moreover, let $\Exp(h) = \Op_\eps^{-1}\left[ \exp\left(\Op_\eps (h)\right)\right]$ be the ``quantum exponential'' 
(so that $\cG = \Exp(-h)$, for $\lam = 0$) and let
\begin{equation}
\label{semiE}
 \Exp(h)= \Exp_0(h) + \eps\, \Exp_1(h) + \eps^2 \Exp_2(h)  + \cdots
\end{equation}
be its formal semiclassical expansion.
Then,
\begin{subequations}
\label{gexp}
\begin{align}
\label{gexpA}
  &\cG_0 = \frac{1}{\e^h + \lam},
\\
\label{gexpB}
  & \cG_{2n+1} = 0,   \qquad n \geq 0,
\\
\label{gexpC}
  &\cG_{2n} =    - \sum_{m = 0}^{n-1}\, \sum_{k+\ell+m = n}  \frac{\Exp_{2k}(h)\, \#_{2\ell}\, \cG_{2m}}{\e^h + \lam},
  \qquad n \geq 1
\end{align}
\end{subequations}
where $\#_{2\ell}$ are the even terms of the (scaled) Moyal product expansion \eqref{MoyalExpansion}.
\end{proposition}
\begin{proof}
Let $H = \Op_\eps(h)$ and $G = \Op_\eps(\cG) = (\e^H + \lam)^{-1}$.
Then, from the relation 
$$
  (\e^H + \lam) G = G (\e^H + \lam) =  I
$$
and the definition of the Moyal product \eqref{Moyal}, we get
$$
  \left(\Exp(h) + \lam\right) \# \cG = \cG \# \left(\Exp(h) + \lam\right) =  1,
$$
that is
$$
  \frac{\left(\Exp(h) + \lam\right) \# \cG +\cG \# \left(\Exp(h) + \lam\right)}{2} =  1.
$$
Now we substitute in this identity the semiclassical expansions of $\Exp(h)$, $\cG$ and $\#$,
and use the following facts:
\begin{enumerate}
\item[(i)]
$\Exp_n(h) = 0$ for odd $n$ (see e.g.\ Ref.\ \cite{DMR05});
\item[(ii)]
$\#_n$ is symmetric for even $n$ and antisymmetric for odd $n$ (see Eq.\ \eqref{MoyalExpansion}).
\end{enumerate}
After the substitution, equating the coefficients of equal powers of $\eps$ yields
$( \e^h + \lam) \cG_0 = 1$ (i.e.\ Eq.\ \eqref{gexpA}) and 
\begin{equation}
\label{aux1}
   \sum_{2k+2\ell+m' = n}  \Exp_{2k}(h) \#_{2\ell}\, \cG_{m'}  = 0
\end{equation}
for $n \geq 1$.
For $n=1$, we immediately get $\cG_1=0$.
If $n>1$ is odd, the sum in \eqref{aux1} has only terms $m' \leq n$ with $m'$ odd and then, 
by induction, we can conclude that $\cG_n = 0$,  which proves Eq.\ \eqref{gexpB}.
Hence, only terms with even $m'$ survive and, therefore, we put $m' = 2m$ in \eqref{aux1}, which 
can be rewritten as follows:
$$
  \sum_{k+\ell+m = n}  \Exp_{2k}(h) \#_{2\ell}\, \cG_{2m}  = 0, \qquad n \geq 1.
$$
Isolating the term with $m=n$, and using $\Exp_0(h) = \e^h$, yields Eq.\ \eqref{gexpC}.
\end{proof}
Equation \eqref{gexpC} recursively relates the terms of the semiclassical expansion of $\cG$ 
to $\cG_0$ (which is a ``classical'' distribution) and to the terms of the expansion of $\Exp(h)$, 
which are well known in literature (see e.g.\ Ref.\ \cite{DMR05}).
\begin{remark}
The expansion of $\Exp(h)$ can obtained as follows. 
For $\beta \geq 0$ let $F(\beta) = \e^{\beta H}$, with $H = \Op_\eps(h)$. 
Then $F(\beta)$ satisfies the semigroup equation
$$
  F'(\beta) = H F(\beta), \qquad F(0) = I,
$$
and, therefore, $f(\beta) = \Op_\eps^{-1}\left(F(\beta)\right) = \Exp(\beta h)$ satisfies
$$
 f'(\beta) = h \# f(\beta), \qquad  f(0) = 1,
$$
whose solution can be easily expanded at the different orders in $\eps$, yielding the expansion of 
$\Exp(h)$ for $\beta = 1$ (see Ref.\ \cite{DMR05} for details).
A similar procedure could be used to get the expansion of $\cG$ from an evolution equation.
Indeed, if we now define $F(\beta) = (\e^{\beta H} + \lam)^{-1}$, it is not difficult to see that
$F$ satisfies the nonlinear semigroup equation
$$
   F'(\beta) = - H F(\beta) (I-\lam F(\beta)), \qquad F(0) = (1+\lam)^{-1}I,
$$
and $f(\beta) = \Op_\eps^{-1}\left(F(\beta)\right)$ satisfies
$$
   f'(\beta) = - h \# f(\beta) \# (1-\lam f(\beta)), \qquad f(0) = (1+\lam)^{-1},
$$
whose semiclassical expansion yields the terms $\cG_n$ for $\beta = 1$.
Note that in the MB case, $\lam = 0$, we get the equation for $\Exp(-h)$.
Note also that in the Bose-Einstein case, $\lam = -1$ the initial datum is singular.
Of course, using Eqs.\ \eqref{gexp} is a much simpler approach (if the expansion of $\Exp(h)$ is known) 
but the differential approach may be of some interest, e.g.\ from the numerical point of view.
\end{remark}
The expansion \eqref{gexp} will now be used to obtain semiclassically approximated expressions
for $A$ and $B$ as functions of $n$ and $J$ from the constraints \eqref{QHDb}, that
we rewrite here:
\begin{equation}
\label{constr}
  \bk{\cG} = n, \qquad \bk{p\,\cG} = J.
\end{equation}
We remark that the expansion \eqref{gexp} refers to $\cG$ as a function of $\eps$, because 
we provisionally assumed $A$ and $B$ of order 1.
We have now to consider that $A$ and $B$ have an expansion in powers of $\eps$, which is determined
from the constraint equations according to the following lemma.
\begin{lemma}
\label{lemmaTaylor}
Let $A$ and $B$ be solutions of the constraint system \eqref{constr}.
Then, they can be formally expanded as follows:
\begin{equation}
\label{semiA}
   A = A^{(0)} + \eps^2 A^{(2)} + \cO(\eps^4),
   \qquad
   B = B^{(0)} + \eps^2 B^{(2)} +\cO(\eps^4),
\end{equation}
where $A^{(0)}$, $A^{(2)}$, $B^{(0)}$ and $B^{(2)}$ satisfy the following system:
\begin{subequations}
\label{expmu}
\begin{align}
\label{expmuA}
&\begin{pmatrix}
  \bk{\cG_0}
\\
  \bk{p_i\,\cG_0}
\end{pmatrix}_{\! (A^{(0)},B^{(0)})} 
= 
\begin{pmatrix}
   n \\ J_i
\end{pmatrix}
\\[6pt]
\label{expmuB}
&\begin{pmatrix}
   \frac{\pt \bk{\cG_0}}{\pt A}
   &
   \frac{\pt \bk{\cG_0}}{\pt B_j}
\\[3pt]
   \frac{\pt \bk{p_i \cG_0}}{\pt A}
   &
   \frac{\pt \bk{p_i \cG_0}}{\pt B_j}
\end{pmatrix}_{\! (A^{(0)},B^{(0)})}
\begin{pmatrix}
   A^{(2)} \\ B^{(2)}_j
\end{pmatrix}
= 
- \begin{pmatrix}
  \bk{\cG_2} \\   \bk{p_i \cG_2}
\end{pmatrix}_{\! (A^{(0)},B^{(0)})}
\end{align}
\end{subequations}
where $\cG_0$ and $\cG_2$ are given by Proposition \ref{keyprop} and 
the subscript $(A^{(0)},B^{(0)})$ means that the expression has to be evaluated in $A = A^{(0)}$ and
 $B = B^{(0)}$. 
\end{lemma}
\begin{proof}
See Appendix \ref{proof_lemmaTaylor}.
\end{proof}
Equations \eqref{expmu} involve moments of the functions $\cG_0$ and $\cG_2$.
Such moments will be explicitly expressed in terms of the functions 
$$
 \phi_s(z) := - \frac{1}{\lam} \Li_s (-\lam \e^z)
$$
($\Li_s$ denoting the polylogarithm function of order $s$ \cite{Lewin81}), which are 
extensively described in Appendix \ref{polylog}.
\par
We begin with the computation of $A^{(0)}$ and $B^{(0)}$, which are determined by 
Eq.\ \eqref{expmuA} alone. 
\begin{proposition}
\label{mu0}
Let $n_d := (2\pi T)^\frac{d}{2}$ and assume 
\begin{equation}
\label{main_assum}
0 < n < \left\{ 
\begin{aligned}
&\frac{n_d\,\zeta\!\left(\frac{d}{2}\right)}{\abs{\lam}},
&\quad &\text{if $\lam < 0$ and $d \geq 3$},
\\
 &\infty,  &\quad &\text{otherwise}
\end{aligned}
\right.
\end{equation}
(where $\zeta$ is the Riemann zeta function).
Then, the solution of system \eqref{expmuA} is
\begin{equation}
\label{A0B0}
 A^{(0)} = T \phi_{\frac{d}{2}}^{-1} \! \left(\frac{n}{n_d} \right),
 \qquad
 B_i^{(0)} = u_i\,,
\end{equation}
where $u = J/n$, and $\phi_{\frac{d}{2}}^{-1}$ is the inverse
of the function $\phi_{\frac{d}{2}}$ (Definition \ref{DefiA1}).
\end{proposition}
\begin{proof}
Since
$$
  \cG_0 = \left(\e^{ \frac{\abs{p-B}^2}{2T} - \frac{A}{T}} + \lam\right)^{-1},
$$
by using Eq.\  \eqref{phint} we obtain
\begin{equation}
\label{momG0}
  \bk{\cG_0} = n_d\, \phi_\frac{d}{2} \! \left(\frac{A}{T}\right),
  \qquad
  \bk{p_i\cG_0}  = B_i \bk{\cG_0}.
\end{equation}
Then, from Eq.\ \eqref{expmuA} we immediately get $B_i^{(0)} = J_i/n = u_i$, 
while for $A^{(0)}$ we have to solve the equation
\begin{equation*}
\label{AUXmom0}
  \phi_\frac{d}{2} \! \left(\frac{A^{(0)}}{T}\right) = \frac{n}{n_d}.
\end{equation*}
Now, $\phi_\frac{d}{2}(z)$ is an increasing function of $z$ (as it is apparent from Eq.\ \eqref{phint}),
and ranges from $0$ to $+\infty$ unless $\lam < 0 $ and $d \geq 3$, in which case
it reaches a maximum value $\zeta\!\left(\frac{d}{2}\right)/\abs{\lam}$ as $z \to 0^-$
(this follows from Eq.\ \eqref{RZ}) .
Thus, in the assumption \eqref{main_assum}, the above equation can be uniquely solved
and we can write $A^{(0)} =  T \phi_{\frac{d}{2}}^{-1} \big(\frac{n}{n_d} \big)$.
\end{proof}
\begin{remark}
\label{rem_condens}
The condition $n < n_d\,\zeta\big(\frac{d}{2}\big)/\abs{\lam}$, for $\lam < 0$ and $d \geq 3$, reflects the fact
that, at dimension 3 or higher, BE statistics is able to ``accommodate'' only a limited number of particles.
The exceeding particles, according to Bose-Einstein theory (and to experiments as well), are expected to 
fall in the fundamental state, giving rise to the Bose-Einstein condensate.
Since $n_d = (2\pi T)^\frac{d}{2}$, we have that the particles are all non-condensate if $T$ is
above the {\em critical temperature}
$$
  T_c = \frac{1}{2\pi}\left( \frac{\abs{\lam}\, n}{\zeta\!\left(\frac{d}{2}\right)}\right)^{2/d},
$$
or, using the physical density $N_0 n$ (with $N_0$ given by \eqref{N0def}),
$$
  T_c = \frac{2\pi \hbar^2}{m k_B} \left( \frac{\abs{\lam}\,n}{\zeta\!\left(\frac{d}{2}\right)}\right)^{2/d}.
$$
Our discussion is therefore limited to the non-condensate, or supercritical, phase.
The full description of a quantum fluid equation with BE statistics would require a coupling between the 
non-condensate and the condensate phases, which is matter for future work.
\end{remark}
Next, we compute $A^{(2)}$ and $B^{(2)}$ from Eq.\ \eqref{expmuB}.
This involves the computation of $\bk{\cG_2}$ and $\bk{p_i \cG_2}$, which is done in next lemma.
\begin{lemma}
The moments  $\bk{\cG_2}$ and $\bk{p_i \cG_2}$, for generic Lagrange multipliers $A$ and $B$,
are given by  
\label{lemmaG2}
\begin{subequations}
\label{momG2}
\begin{align}
\label{momG2A}
&\begin{aligned}
  \bk{\cG_2} &= \frac{n_d}{24 T^2} \left[ 2\Delta A - \frac{\pt B_j}{\pt x_k}  
  \left(\frac{\pt B_j}{\pt x_k} - \frac{\pt B_k}{\pt x_j} \right)  \right] \phi_{\frac{d}{2}-2} \left(\frac{A}{T}\right)
  \\
  &+  \frac{n_d}{24 T^3} \abs{\nabla A}^2 \, \phi_{\frac{d}{2}-3} \left(\frac{A}{T}\right),
  \end{aligned}
\\[6pt]
\label{momG2B}
  &\bk{p_i \cG_2} = B_i \bk{\cG_2} + \frac{n_d}{12 T}\, \ptx{j} 
  \left[ \left( \frac{\pt B_i}{\pt x_j} - \frac{\pt B_j}{\pt x_i} \right)  \phi_{\frac{d}{2}-1} \left(\frac{A}{T}\right)  \right],
\end{align}
\end{subequations}
with $n_d = (2\pi T)^\frac{d}{2}$.
\end{lemma}
\begin{proof}
See Appendix \ref{proof_lemmaG2}.
\end{proof}
Thanks to Lemma \ref{lemmaG2}, we are now ready to compute the second-order terms in the 
semiclassical expansion \eqref{semiA} of the Lagrange multipliers, which are obtained from 
system \eqref{expmuB}.
\begin{proposition}
\label{mu2}
The solution of system \eqref{expmuB} is
\begin{subequations}
\label{A2B2}
\begin{align}
\label{A2}
 &\begin{aligned}
   A^{(2)} &= 
   \frac{1}{24T}  
   \frac{\pt u_j}{\pt x_k} \left(\frac{\pt u_j}{\pt x_k} - \frac{\pt u_k}{\pt x_j} \right) 
   \frac{\phi_{\frac{d}{2}-2}^0(n)}{\phi_{\frac{d}{2}-1}^0(n)}
\\
   &- \frac{1}{24} \left[ 
   \frac{2\Delta A^{(0)}(n)}{T} \,
   \frac{\phi_{\frac{d}{2}-2}^0(n)}{\phi_{\frac{d}{2}-1}^0(n)}
   + \frac{\abs{\nabla A^{(0)}(n)}^2}{T^2} \,
   \frac{\phi_{\frac{d}{2}-3}^0(n)}{\phi_{\frac{d}{2}-1}^0(n)}
    \right]
   \end{aligned}
\\[6pt]
\label{B2}
 &B_i^{(2)} = \frac{n_d}{12 T n}\, \ptx{j} 
  \left[ \left( \frac{\pt u_j}{\pt x_i} - \frac{\pt u_i}{\pt x_j} \right)  \phi_{\frac{d}{2}-1}^0(n)  \right],
\end{align}
\end{subequations}
where $A^{(0)}(n)$  is given by Eq. \eqref{A0B0}, $u = J/n$, $n_d = (2\pi T )^\frac{d}{2}$ and 
\begin{equation}
\label{phi0def}
  \phi_s^0(n) = \phi_s\left(\frac{A^{(0)}(n)}{T} \right) 
  = \phi_s \left(   \phi_{\frac{d}{2}}^{-1}\left(\frac{n}{n_d} \right) \right).
\end{equation}
\end{proposition}
\begin{proof}
What we have to do is solving the linear equation \eqref{expmuB} for the unknowns  
$A^{(2)}$ and  $B^{(2)}$, the expressions of $A^{(0)}$ and $B^{(0)}$ being given by \eqref{A0B0}.
The derivatives of $\bk{\cG_0}$ and $\bk{p_i \cG_0}$ with respect of $A$ and $B$ 
are easily obtained from Eq.\ \eqref{momG0}. 
Evaluating the resulting expressions in $A = A^{(0)}$ and $B = B^{(0)}$ we obtain
$$
\begin{pmatrix}
   \frac{\pt \bk{\cG_0}}{\pt A}
   &
   \frac{\pt \bk{\cG_0}}{\pt B_j}
\\[4pt]
   \frac{\pt \bk{p_i \cG_0}}{\pt A}
   &
   \frac{\pt \bk{p_i \cG_0}}{\pt B_j}
\end{pmatrix}_{\! (A^{(0)},B^{(0)})}
 =  \frac{n_d}{T}
 \begin{pmatrix}
   \phi_{\frac{d}{2}-1}^0  &  0
\\[4pt]
    u_i \phi_{\frac{d}{2}-1}^0  &  \delta_{ij} T n/n_d
\end{pmatrix},
$$
where $\phi_s^0 = \phi_s^0(n)$ is given by \eqref{phi0def}, and we used the fact that 
$\phi_{\frac{d}{2}}^0 = \frac{n}{n_d}$. 
The inverse matrix is easily computed to be
$$
\left[ \frac{n_d}{T}
   \begin{pmatrix}
   \phi_{\frac{d}{2}-1}^0   &   0
\\[4pt]
   u_i \phi_{\frac{d}{2}-1}^0   &  \delta_{ij} Tn/n_d
\end{pmatrix} \right]^{-1}
 =  \begin{pmatrix}
   T/(n_d\,\phi_{\frac{d}{2}-1}^0)   &  0
\\[4pt]
   - u_i/n   &  \delta_{ij} / n
\end{pmatrix}
$$
and then
$$
\begin{pmatrix}
   A^{(2)} \\ B^{(2)}_i
\end{pmatrix}
 = 
   \begin{pmatrix}
   - T/(n_d\,\phi_{\frac{d}{2}-1}^0)   &   0
\\[4pt]
   u_i/n   &  - \delta_{ij} / n
\end{pmatrix}
\begin{pmatrix}
  \bk{\cG_2} \\   \bk{p_j \cG_2}
\end{pmatrix}_{\! (A^{(0)},B^{(0)})},
$$
where $\bk{\cG_2}_{(A^{(0)},B^{(0)})}$ and $\bk{p_j \cG_2}_{(A^{(0)},B^{(0)})}$
are obtained by substituting  $A^{(0)}$ and $B^{(0)}$ for $A$ and $B$ in the 
expressions \eqref{momG2}.
This immediately yields Eqs.\ \eqref{A2B2}.
\end{proof}
By using the derivation rules
\begin{equation}
\label{derules}
  \frac{ \nabla A^{(0)}(n)}{T} = \frac{\nabla n}{n_d\,\phi^0_{\frac{d}{2}-1}(n)},
 \qquad
  \nabla \phi_s^0(n) = \frac{\phi_{s-1}^0(n)}{\phi^0_{\frac{d}{2}-1}(n)} \,\frac{\nabla n}{n_d},
\end{equation}
it is readily seen that the term between square brackets in Eq.\ \eqref{A2} can be given the more 
explicit form
\begin{multline}
\label{Qexplicit}
     \frac{2\Delta A^{(0)}}{T}\,
     \frac{\phi_{\frac{d}{2}-2}^0}{\phi_{\frac{d}{2}-1}^0}
  +   \frac{\abs{\nabla A^{(0)}}^2}{T^2} \,
   \frac{\phi_{\frac{d}{2}-3}^0}{\phi_{\frac{d}{2}-1}^0}
 =  \frac{2 \Delta n}{n_d}\,
 \frac{\phi_{\frac{d}{2}-2}^0}{(\phi_{\frac{d}{2}-1}^0)^2}
  + \frac{\abs{\nabla n}^2}{n_d^2}
  \left( \frac{\phi_{\frac{d}{2}-3}^0}{(\phi_{\frac{d}{2}-1}^0)^3}  
  - 2\frac{(\phi_{\frac{d}{2}-2}^0)^2}{(\phi_{\frac{d}{2}-1}^0)^4} \right)
\end{multline}
(where the arguments $n$ have been omitted).
As we shall see in Sec.\ \ref{S4.2}, this term can be identified as a modified Bohm potential,
since it gives the usual ``statistical'' Bohm potential in the MB limit $\lam\to 0$.
\subsection{Semiclassical fluid equations}
\label{Sec_SFE}
The semiclassical expansion of the Lagrange multipliers, found in the previous section, 
can now be substituted in Eqs. \eqref{QHDa} and \eqref{QDDa} to obtain
semiclassical hydrodynamic and drift-diffusion equations.
\par
In order to do that, let us consider the term \eqref{Qclosure} that appear in both equations 
(with $B=0$ in the diffusive case) and contains the Lagrange multipliers.
Let us rewrite it, by using the velocity variable $u = J/n$, and expand it according to \eqref{semiA}.
Taking account that $B^{(0)} = u$ (Proposition \ref{mu0}), we obtain
\begin{multline*}
  \ptx{j} (n u_i  B_j) +  n\left(u_j - B_j \right)  \frac{\pt B_j}{\pt x_i} + n \frac{\pt A}{\pt x_i}
   = \ptx{j} (n u_i  u_j) + n  \ptx{i} A^{(0)} 
\\
  + \eps^2\left(  n B_j^{(2)} R_{ij}
   + u_i \ptx{j}(n B_j^{(2)}) + n  \ptx{i} A^{(2)} \right) + \cO(\eps^4),
\end{multline*}
where we introduced the notation 
\begin{equation}
\label{Rdef}
   R_{ij} =  \frac{\pt u_i}{\pt x_j} - \frac{\pt u_j}{\pt x_i}  
\end{equation}
for the velocity curl tensor.
Now, using Eq. \eqref{derules} we can write 
$$
    \ptx{i} A^{(0)} = \frac{T}{n_d\,\phi^0_{\frac{d}{2}-1}}\,\frac{\pt n}{\pt x_i}.
$$
Moreover, from Eqs.\ \eqref{A2B2} and \eqref{Qexplicit} we obtain
\begin{multline*}
n B_j^{(2)} R_{ij} + u_i \ptx{j}(n B_j^{(2)}) + n \ptx{i} A^{(2)}   = 
\\
\frac{n_d  R_{ij} }{12 T}\, \ptx{k} 
  \left( R_{kj} \phi_{\frac{d}{2}-1}^0 \right)
  +  \frac{n}{48 T}  \ptx{i}  \left(  R_{jk} R_{jk}\, \frac{\phi_{\frac{d}{2}-2}^0 }{\phi_{\frac{d}{2}-1}^0} \right)
  +  n \ptx{i} Q(n)
\end{multline*}
where we used the identities
$$
 R_{jk} \frac{\pt u_j}{\pt x_k}  = \frac{1}{2} R_{jk} R_{jk},
 \qquad
 \frac{\pt^2}{\pt x_j \pt x_k} \left( R_{kj} \phi_{\frac{d}{2}-1}^0 \right) = 0,
$$
and introduced the notation 
\begin{equation}
\label{Qdef}
 Q(n) = -\frac{1}{24} \left[ \frac{2 \Delta n}{n_d}\,
 \frac{\phi_{\frac{d}{2}-2}^0}{(\phi_{\frac{d}{2}-1}^0)^2}
  + \frac{\abs{\nabla n}^2}{n_d^2}
  \left( \frac{\phi_{\frac{d}{2}-3}^0}{(\phi_{\frac{d}{2}-1}^0)^3}  
  - 2\frac{(\phi_{\frac{d}{2}-2}^0)^2}{(\phi_{\frac{d}{2}-1}^0)^4}\right) \right]
\end{equation}
for the modified Bohm potential.
Hence, we can state the main results of this section.
\begin{theorem}
\label{main_thm_1}
Assume that condition \eqref{main_assum} is satisfied for all times.
Then, neglecting terms of order $\cO(\eps^4)$, the isothermal quantum hydrodynamic model \eqref{QHD} 
admits the following, formal, approximation 
\begin{equation}
\label{semi_hydro}
\left\{ 
\begin{aligned}
 &\frac{\pt n}{\pt t} +  \ptx{i} (nu_i) = 0
 \\
 &\ptt (nu_i) +   \ptx{j} (n u_i  u_j) + n\frac{\pt V}{\pt x_i} 
 +  \frac{Tn}{n_d\,\phi^0_{\frac{d}{2}-1}}\frac{\pt n}{\pt x_i} 
    +  \eps^2n \frac{\pt Q}{\pt x_i}
 \\
  &\qquad \quad
  + \frac{\eps^2 n_d R_{ij} }{12 T}\, \ptx{k}  \left( R_{kj} \phi_{\frac{d}{2}-1}^0 \right)
  +  \frac{\eps^2 n}{48 T}  \ptx{i}  
  \left( R_{jk} R_{jk} \,\frac{\phi_{\frac{d}{2}-2}^0 }{\phi_{\frac{d}{2}-1}^0} \right)
  = 0,
\end{aligned}
\right.
\end{equation}
where $\phi_s^0 = \phi_s^0(n)$ is given by Eq.\ \eqref{phi0def}, $R_{ij}$ is given
by  Eq.\ \eqref{Rdef}, $Q = Q(n)$ is given by  Eq.\ \eqref{Qdef} and  $n_d = (2\pi T)^\frac{d}{2}$.
\end{theorem}
\begin{theorem}
\label{main_thm_2}
Assume that condition \eqref{main_assum} is satisfied at all times.
Then, neglecting terms of order $\cO(\eps^4)$, the quantum diffusive model \eqref{QDD} 
admits the following, formal, approximation 
\begin{equation}
\label{semi_diffu}
  \frac{\pt n}{\pt t} =  \ptx{i} \left(  \frac{Tn}{n_d\,\phi^0_{\frac{d}{2}-1}}\frac{\pt n}{\pt x_i}
   + n\frac{\pt V}{\pt x_i} + \eps^2 n\frac{\pt Q}{\pt x_i}\right),
\end{equation}
where $\phi_s^0 = \phi_s^0(n)$ is given by Eq.\ \eqref{phi0def}, 
$Q = Q(n)$ is given by  Eq.\ \eqref{Qdef} and  $n_d = (2\pi T)^\frac{d}{2}$.
\end{theorem}
Equations \eqref{semi_hydro} and \eqref{semi_diffu} are the generalized version of the 
semiclassical hydrodynamic and diffusive equations  derived for MB statistics
in Refs.\ \cite{DGM07,JM05} and \cite{DMR05} (see also Subsec.\ \ref{S4.2}).
Equation \eqref{semi_diffu} with $\eps = 0$ has been derived in Ref.\ \cite{JKP11} 
(where $\lam >0$ is assumed, although this does not affects the form of the equation).
A simplified version of Eq.\  \eqref{semi_hydro} has been derived in Ref. \cite{TR10}, where
the terms of order $\eps^2$ that depend on $R$ are missing.

\section{Analysis of particular regimes}
\label{S4}
In this section we investigate the form taken by the semiclassical hydrodynamic and diffusive 
equations, Eqs.\ \eqref{semi_hydro} and \eqref{semi_diffu}, in some specific physical regime. 
In particular, we shall consider the irrotational regime, the Maxwell-Boltzmann limit, 
$\lam \to 0$, and the zero temperature limit, $T\to 0$.  
Let us stress the fact that all statement and proofs are purely formal.
\subsection{The irrotational fluid}
\label{S4.1}
First of all, let us look at the form taken by the hydrodynamic equations \eqref{semi_hydro}
when initial data are irrotational ($R = 0$).
The following proposition and its proof are similar to those of an analogous result given in 
Ref.\ \cite{DGM07} for the fully-quantum equations with general convex entropy.
On the other hand, our statement can be a bit more precise because we are dealing 
with the simpler case of semiclassical approximation.
\begin{proposition}
\label{irro}
Let $(n, u)$ be a smooth solution of Eq.\ \eqref{semi_hydro} such that $B^{(2)} \in W^{1,\infty}(\mR^d)$
at all times, and assume that the fluid is initially irrotational, i.e.\ $R = 0$ at $t= 0$. 
Then the fluid remains irrotational at all times and, therefore, $n$ and $u$ satisfy
\begin{equation}
\label{semi_hydro_irro}
\left\{ 
\begin{aligned}
 &\frac{\pt n}{\pt t} +  \ptx{i} (nu_i) = 0
 \\
 &\frac{\pt u_i}{\pt t}  +  u_j \frac{\pt u_j}{\pt x_i}  + \frac{\pt V}{\pt x_i} 
 +  \frac{T}{n_d\,\phi^0_{\frac{d}{2}-1}}\frac{\pt n}{\pt x_i} 
    +  \eps^2 \frac{\pt Q}{\pt x_i}  = 0.
\end{aligned}
\right.
\end{equation}
\end{proposition}
\begin{proof}
By using the identity
\begin{equation}
\label{irrid}
  \ptt (nu_i) +   \ptx{j} (n u_i  u_j) = n \left( \frac{\pt u_i}{\pt t} + R_{ij}u_j + u_j \frac{\pt u_j}{\pt x_i} \right) 
\end{equation}
(where the continuity equation, i.e.\ the first of \eqref{semi_hydro}, was used), the second equation of \eqref{semi_hydro}
can be rewritten as follows:
\begin{multline*}
  \frac{\pt u_i}{\pt t} + R_{ij}u_j  + \frac{\eps^2 n_d R_{ij}}{12 T n}\, \ptx{l} \left( R_{lj}  \phi_{\frac{d}{2}-1}^0  \right)
\\
  + \ptx{i} \left( \frac{1}{2}\abs{u}^2 + V + T\phi_{\frac{d}{2}}^{-1}\Big(\frac{n}{n_d}\Big)
  + \eps^2 Q  + \frac{\eps^2 R_{jk} R_{jk}}{48 T} \,\frac{\phi_{\frac{d}{2}-2}^0 }{\phi_{\frac{d}{2}-1}^0} \right) = 0.
\end{multline*}
Then, the following equation for $R_{ik} =  \frac{\pt u_i}{\pt x_k} - \frac{\pt u_k}{\pt x_i}$ is obtained:
$$
 \ptt R_{ik} + \ptx{k} \left(R_{ij} \theta_j \right) - \ptx{i} \left(R_{kj} \theta_j \right) = 0,
$$
where
$$
\theta_j := u_j  +  \frac{\eps^2 n_d}{12 T n}\, \ptx{l} \left( R_{lj}  \phi_{\frac{d}{2}-1}^0  \right)
= B_j^{(0)} + \eps^2 B_j^{(2)} ,
$$
and also (from a direct calculation using the definition of $R_{ij}$)
$$
 \ptt R_{ik}  + \theta_j \ptx{j} R_{ik} + R_{ij} \ptx{k} \theta_j - R_{kj} \ptx{i} \theta_j  = 0.
$$
Multiplying both sides by $R_{ik}$, summing over $i$ and $k$, and integrating over $x \in \mR^d$ yields
$$
   \frac{d}{dt} \int_{\mR^d} \abs{R}^2 dx =
   \int_{\mR^d}  \abs{R}^2 \DIV \theta\, dx + 4  \int_{\mR^d} R^2 : \nabla \theta \,dx,
$$
where $\abs{R}^2 = \sum_{i,k} R_{ik}^2$ and $R^2 : \nabla \theta = R_{ik} R_{kj} \ptx{i}\theta_j$.
From our assumptions on $u$ and $B^{(2)}$ we have that $\norma{\nabla \theta}_\infty$ is 
finite (and continuous) in time, and we can write
$$
   \frac{d}{dt} \int_{\mR^d} \abs{R}^2 dx \leq C \norma{\nabla \theta}_\infty \int_{\mR^d}  \abs{R}^2 dx
$$
for some constant $C>0$.
Since $R=0$ at $t=0$,  by the Gronwall lemma we obtain that $R=0$ at all times and then,
using again the identity \eqref{irrid}, Eq.\ \eqref{semi_hydro} reduces to Eq.\ \eqref{semi_hydro_irro}.
\end{proof}
We remark that an important class of irrotational initial data is that of pure states.
In fact, it is easy to show that the velocity field
$$
 u  = \frac{\eps}{2i \abs{\psi}^2} \left( \overline \psi  \nabla \psi
- \psi \nabla  \overline \psi \right),
$$
associated to the pure state represented by the wave function $\psi(x)$ (see Eq.\ \eqref{JDef}), has $R=0$.
Of course, Proposition \ref{irro} is not saying that an initially pure state will remain pure, but just
that it will remain irrotational. 
In Ref.\ \cite{DGM07} it is proven that a pure state remains pure in the limit $T \to 0$ and assuming MB statistics,
in which case the fully-quantum system \eqref{QHD} reduces to Madelung equations \eqref{MadEq} (that are the 
hydrodynamic form of Schr\"odinger equation \cite{Madelung26}).
We stress that is the  fully-quantum system \eqref{QHD} that possesses a limit for $T\to0$ and 
{\em not} the semiclassical equations Eq.\ \eqref{semi_hydro}, which behave singularly in this limit  for 
MB statistics.
We shall discuss this point with more details in Subsec.\ \ref{S4.3}.
\subsection{The Maxwell-Boltzmann limit}
\label{S4.2}
The MB limit of Eqs.\ \eqref{semi_hydro} and \eqref{semi_diffu} is obtained by letting $\lam \to 0$.
We remark that the parameter $\lam$ is hidden in the functions $\phi^0_s$, that are defined
by \eqref{phi0def} and \eqref{phiDef}.
Then, from property \eqref{phi_prop1} we immediately obtain 
\begin{equation}
\label{phiMBlimit}
 \lim_{\lam \to 0} \phi_s^0(n) = \frac{n}{n_d},
\end{equation}
for any $s \in \mR$.
In particular, as far as the modified Bohm potential is concerned (see Eq.\ \eqref{Qdef}),
we obtain
\begin{equation}
\label{Qlim}
 \lim_{\lam \to 0} \eps^2 Q(n) = 
 -\frac{\eps^2}{24} \left( \frac{2 \Delta n}{n} - \frac{\abs{\nabla n}^2}{n^2} \right)
 = -\frac{\eps^2}{6} \frac{\Delta \sqrt{n}}{\sqrt{n}},
\end{equation}
that is the usual (statistical) Bohm potential \cite{DMR05}.
\begin{remark}
\label{Bohm}
What is commonly termed ``Bohm potential'' \cite{Bohm52a,Bohm52b} is the quantum potential appearing in 
Madelung equations \eqref{MadEq}, namely
$$
  V_B(n) = -\frac{\eps^2}{2} \frac{\Delta \sqrt{n}}{\sqrt{n}}.
$$
This is a ``pure-state'' Bohm potential, which differs for a factor $1/3$ from what we termed 
``statistical'' Bohm potential. i.e.\ \eqref{Qlim}.
The latter arises naturally from the quantum entropy principle. 
Which form the Bohm potential should have in quantum fluid equations
is a long-standing debate, see e.g.\  Ref.\ \cite{FerryZhou93} and references therein.
\end{remark}
As far as the rotational terms are concerned, i.e.\ the terms of Eq.\ \eqref{semi_hydro}
that depend on the velocity curl tensor $R$, we obtain
\begin{multline*}
  \lim_{\lam \to 0} \left[ \frac{\eps^2 n_d R_{ij} }{12 T}\, \ptx{k}  \left( R_{kj} \phi_{\frac{d}{2}-1}^0 \right)
  +  \frac{\eps^2 n}{48 T}  \ptx{i}  
  \left( R_{jk} R_{jk} \,\frac{\phi_{\frac{d}{2}-2}^0 }{\phi_{\frac{d}{2}-1}^0} \right) \right]
 \\[3pt]
  =  \frac{\eps^2 R_{ij} }{12 T}\, \ptx{k}  \left( R_{kj} n\right)
  +  \frac{\eps^2n}{48 T}  \ptx{i}    \left( R_{jk} R_{jk} \right).
\end{multline*}
The last expression can be simplified by considering the identity
$$
  \ptx{k} \left(nR_{ij}R_{kj} \right) = R_{ij} \ptx{k} \left(nR_{kj} \right)
  + \frac{n}{4} \ptx{i}\left( R_{jk} R_{jk} \right) 
$$
(recall that we sum over the repeated indices $j$ and $k$), so that
$$
  \frac{\eps^2 R_{ij} }{12 T}\, \ptx{k}  \left( R_{kj} n\right)
  +  \frac{\eps^2n}{48 T}  \ptx{i}    \left( R_{jk} R_{jk} \right)
  = \frac{\eps^2}{12 T} \ptx{k} \left(nR_{ij}R_{kj} \right),
$$
where $\ptx{k} \left(nR_{ij}R_{kj} \right)$ is the expression in components
of $\DIV(n R R^T)$
(this is the form in which the rotational terms are written in Ref.\ \cite{JM05}).
All this considered we can state the following.
\begin{proposition}
In the Maxwell-Boltzmann limit, $\lam \to 0$, the semiclassical hydrodynamic and diffusive
equations \eqref{semi_hydro} and \eqref{semi_diffu} take, respectively,  the  form
\begin{equation}
\label{MB_hydro}
\left\{ 
\begin{aligned}
 &\frac{\pt n}{\pt t} +  \ptx{i} (nu_i) = 0
 \\[4pt]
 &\ptt (nu_i) +   \ptx{j} (n u_i  u_j) + n\frac{\pt V}{\pt x_i} 
 +  T \frac{\pt n}{\pt x_i} 
\\
  &\qquad \qquad
   - \frac{\eps^2}{6} n \ptx{i} \frac{\Delta \sqrt{n}}{\sqrt{n}}
   + \frac{\eps^2}{12 T} \ptx{k} \left(nR_{ij}R_{kj} \right) = 0,
\end{aligned}
\right.
\end{equation}
and
\begin{equation}
\label{MB_diffu}
  \frac{\pt n}{\pt t} =  \ptx{i} \left(  T \frac{\pt n}{\pt x_i}
   + n\frac{\pt V}{\pt x_i} - \frac{\eps^2}{6} n \ptx{i} \frac{\Delta \sqrt{n}}{\sqrt{n}} \right).
\end{equation}
\end{proposition}
The isothermal hydrodynamic equations \eqref{MB_hydro} were first derived in Refs.\ \cite{JM05} 
and \cite{DGM07}.
In the latter, the rotational term is expressed in the equivalent%
\footnote{Actually, in Ref.\ \cite{DGM07} the factor $1/T$ seems to be missing.}
form (in the three-dimensional case)
$$
  \frac{\eps^2}{12 T} \DIV \left(n R R^T \right) = 
  \frac{\eps^2}{12 T} \,\omega \times \left( \nabla \times (n\omega) \right)
  + \frac{\eps^2}{24 T}\, n \nabla \abs{\omega}^2,
$$ 
where $\omega = \nabla \times u$.
The diffusive equation \eqref{MB_diffu} was first derived in Ref.\ \cite{DMR05}.
\par
We finally remark that, as it can be easily deduced from \eqref{phi_prop1.2}, the MB limit can be equivalently 
obtained by fixing $\lam \not= 0$ and letting $T \to +\infty$.
\subsection{The zero-temperature limit}
\label{S4.3}
The behavior of Eqs.\ \eqref{semi_hydro} and \eqref{semi_diffu} in the limit $T\to 0$  depends dramatically 
on the sign of $\lam$. 
For this reason we divide the analysis of such limit in the three reference cases $\lam = 1$ (FD),  
$\lam = 0$ (MB) and $\lam = -1$ (BE).
\par
\smallskip
In view of the following discussion, it is convenient to extend property \eqref{phi_prop2} to
all positive $\lam$.
For $z \in \mR$ and $s \in \mR$, let us denote by $F_s(z) := - \Li_s(-\e^z)$ 
the function $\phi_s(z)$ for $\lam = 1$.
Then, assuming $\lam > 0$, from property \eqref{phi_prop2} we obtain
\begin{equation}
\label{prop2new}
\phi_s(z)  = 
\frac{1}{\lam} F_s(z + \log\lam) \sim \frac{(z+ \log\lam)^s}{\lam \, \Gamma(s+1)},
\quad 
\text{as $z \to +\infty$},
\end{equation}
where $s \not= -1, -2, \ldots$, and $f \sim g$ means $f/g \to 1$.
Then, recalling that $n_d = (2\pi T)^\frac{d}{2}$, we have
\begin{equation}
\label{invapprox}
\phi_{\frac{d}{2}}^{-1}\left(\frac{n}{n_d} \right)  \sim 
\left(\frac{\lam\, \Gamma\big(\frac{d}{2} + 1\big) n}{n_d} \right)^\frac{2}{d} - \log \lam,
\quad 
\text{as $T \to 0$}.
\end{equation}
Hence, recalling definition \eqref{phi0def} and combining \eqref{prop2new} with \eqref{invapprox}
we obtain
\begin{equation}
\label{phi0approx}
\phi^0_s(n) \sim 
\frac{\lam^{\frac{2s}{d}-1}}{\Gamma(s+1)}
\left(\frac{\Gamma\big(\frac{d}{2} + 1\big) n}{n_d} \right)^\frac{2s}{d},
\quad 
\text{as $T \to 0$},
\end{equation}
which holds for $\lam > 0$ and $s \not= -1, -2, \ldots$.
This formula can be used to obtain the formal asymptotics for $T \to 0$ 
of the various  temperature-dependent terms
that appear in Eqs.\ \eqref{semi_hydro} and \eqref{semi_diffu}. 
In particular, we have
\begin{subequations}
\label{Tlimits}
\begin{align}
\label{Tlimits1}
   &\frac{T}{n_d\,\phi^0_{\frac{d}{2}-1}} \to
   \frac{\lam^\frac{2}{d} \left(\frac{d}{2}\right)^{\frac{2-d}{d}} 
   \Gamma\big(\frac{d}{2}\big)^\frac{2}{d}\, n^\frac{2-d}{d}}{2\pi},
\\[4pt]
\label{Tlimits2}
  &\frac{1}{n_d}\frac{\phi_{\frac{d}{2}-2}^0}{(\phi_{\frac{d}{2}-1}^0)^2}
    \to \frac{(d-2)}{d} \,\frac{1}{n}
\\[4pt]
\label{Tlimits3}
  &\frac{1}{n_d^2}\frac{\phi_{\frac{d}{2}-3}^0}{(\phi_{\frac{d}{2}-1}^0)^3}
    \to \frac{(d-4)(d-2)}{d^2} \,\frac{1}{n^2}
\\[4pt]
\label{Tlimits5}
  &\frac{1}{T}\,\frac{\phi_{\frac{d}{2}-2}^0}{\phi_{\frac{d}{2}-1}^0}
    \to  \frac{d-2}{2}\, \frac { 2\pi \,n^{-\frac{2}{d}} }
   { \lam^\frac{2}{d}\, \Gamma\big(\frac{d}{2}+1\big)^\frac{2}{d} },
\end{align}
\end{subequations}
as $T\to 0$ (note that the right-hand sides do not depend on $T$).
From  \eqref{Tlimits2}, \eqref{Tlimits3} and \eqref{Qdef} we get
the interesting limit
\begin{equation}
\label{Tlimits6}
  Q(n) \to  - \frac{1}{6}\, \frac{(d-2)}{d} \,  \frac{\Delta \sqrt{n}}{\sqrt{n}},
  \quad 
\text{as $T \to 0$},
\end{equation}
which is independent on $\lam$ too.
\subsubsection{FD case}
Let us first of all consider the limit $T\to 0$ of Eqs.\ \eqref{semi_hydro} and \eqref{semi_diffu}
assuming  FD statistics. 
This is the richest case since, as Eqs.\ \eqref{Tlimits} and \eqref{Tlimits6} show, for $\lam > 0$
the behavior of the semiclassical fluid equations is regular as temperature goes to 0
(this is not the case for MB and BE statistics, as we shall see next).
Then, it is enough to set $\lam = 1$ in Eqs.\  \eqref{Tlimits} and \eqref{Tlimits6} to obtain the following.
\begin{proposition}
Let $\lam = 1$. Then, in the limit $T \to 0$ (also known as completely degenerate limit)
the semiclassical hydrodynamic and diffusive
equations \eqref{semi_hydro} and \eqref{semi_diffu} take, respectively,  the  form
\begin{equation}
\label{T0_FD_hydro}
\left\{ 
\begin{aligned}
 &\frac{\pt n}{\pt t} +  \ptx{i} (nu_i) = 0
 \\[4pt]
 &\ptt (nu_i) +   \ptx{j} (n u_i  u_j) + n\frac{\pt V}{\pt x_i} 
 + \gamma_1 \ptx{i} n^\frac{2+d}{d}     
 - \eps^2 \gamma_2\, n\ptx{i} \frac{\Delta \sqrt{n}}{\sqrt{n}}
\\[2pt]
  &\qquad \qquad 
        + \eps^2 \gamma_3\, R_{ij}  \ptx{k}  \left( R_{kj} n^\frac{d-2}{d} \right)
  +  \eps^2 \gamma_4\, n \ptx{i} \frac{R_{jk} R_{jk}}{n^\frac{2}{d} }
 = 0,
\end{aligned}
\right.
\end{equation}
and
\begin{equation}
\label{T0_FD_diffu}
  \frac{\pt n}{\pt t} =  \ptx{i} \left(  \gamma_1 \ptx{i} n^\frac{2+d}{d} 
   + n\frac{\pt V}{\pt x_i} - \eps^2 \gamma_2\, n \ptx{i} \frac{\Delta \sqrt{n}}{\sqrt{n}} \right),
\end{equation}
where
\begin{equation*}
\begin{aligned}
&\gamma_1 = \frac{1}{2\pi}\, \frac{d}{d+2}\, \Big(\frac{d}{2} \Big)^\frac{2-d}{d} 
   \Gamma\Big(\frac{d}{2}\Big)^\frac{2}{d}, 
&\quad &\gamma_2 = \frac{d-2}{6 d},
\\[6pt]
&\gamma_3 =  \frac{d \pi}{12\,\Gamma\big(\frac{d}{2}+ 1\big)^{\frac{2}{d} }}\, ,
&  &\gamma_4 =  \frac{(d-2) \pi}{48\,\Gamma\big(\frac{d}{2} + 1\big)^{\frac{2}{d} }} .
\end{aligned}
\end{equation*}
\end{proposition}
We remark the particularly simple form of the limit Bohm potential: it is just the (statistical) Bohm potential
multiplied by the factor $\frac{d-2}{d}$. 
Noticeably, it vanishes for $d=2$ and changes sign for $d=1$. Note that for $d=2$ also the coefficient $\gamma_4$
vanishes.
\par
Equations \eqref{T0_FD_hydro}  and  \eqref{T0_FD_diffu} with $R=0$ and $d=3$ have been obtained 
in Ref.\ \cite{TR10} (where also  ``weakly degenerate'' and ``strongly degenerate'' limits are considered). 
Equation  \eqref{T0_FD_diffu} with $\eps = 0$ and $d=3$ has been obtained in Ref.\ \cite{JKP11}
(where also energy-transport equations are considered, which however reduce to the diffusive equation 
for $T\to 0$).
Further references to degenerate fluid models can be found in Refs.\ \cite{JKP11,TR10,TR11}.
\subsubsection{MB case}
Maxwell-Boltzmann stastistics, for the $T\to 0$ limit, is a very singular case.
Indeed, looking at Eqs. \eqref{phiMBlimit} and  \eqref{phi0approx}, we notice that, at fixed $n$,
the function $\phi^0_{\frac{d}{2}-k}(n)$ (where $k = 1,2,3$ are the relevant cases) behaves like $T^{-\frac{d}{2}}$
for $\lam \to 0$ and  $T > 0$, and behaves like $\lam^{-\frac{2k}{d}} T^{k - \frac{d}{2}}$ for $T \to 0$ and
$\lam > 0$. 
Then, the two limits $\lam \to 0$ and $T\to 0$ are somehow incompatible, and the corresponding 
behavior of the fluid equations depends on how the point $(0,0)$ is approached in the parameter 
space $(\lam, T)$.
\par
Let us consider the two paths: $\lam \to 0$ followed by $T \to 0$,
and  $T \to 0$ followed $\lam \to 0$.
The first path corresponds  to starting from 
the MB equations \eqref{MB_hydro} and  \eqref{MB_diffu}, and then letting $T \to 0$;
the second path corresponds to using first the asymptotic identities 
\eqref{Tlimits} and \eqref{Tlimits6} in Eqs.\ \eqref{semi_hydro} and \eqref{semi_diffu}, 
and then letting $\lam \to 0$.
In both cases the rotational terms are singular and, therefore, the limit is only
compatible with irrotational solutions (see Proposition \ref{irro}).
Moreover, in both cases, the diffusive term $\frac{Tn}{n_d\,\phi^0_{d/2-1}}\frac{\pt n}{\pt x_i}$ 
vanishes asymptotically.
Thus, assuming $R=0$ (see Eq.\ \eqref{semi_hydro_irro}), we obtain from both paths equations of the form
\begin{equation}
\label{T0MB_hydro}
\left\{ 
\begin{aligned}
 &\frac{\pt n}{\pt t} +  \ptx{i} (nu_i) = 0
 \\[4pt]
 &\frac{\pt u_i}{\pt t}  +  u_j \frac{\pt u_j}{\pt x_i} + \frac{\pt V}{\pt x_i} 
   - \gamma\, \frac{\eps^2}{6}  \ptx{i} \frac{\Delta \sqrt{n}}{\sqrt{n}}
= 0,
\end{aligned}
\right.
\end{equation}
and 
\begin{equation}
\label{T0MB_diffu}
  \frac{\pt n}{\pt t} =  \ptx{i} \left(  n\frac{\pt V}{\pt x_i} -  \gamma\, \frac{\eps^2}{6} n \ptx{i} \frac{\Delta \sqrt{n}}{\sqrt{n}} \right),
\end{equation}
but the coefficient $\gamma$ changes: it is 1 for the first path and $\frac{d-2}{d}$ for the second path.
 \par
As already mentioned, the correct point of view is probably that of Ref.\ \cite{DGM07},
where the $T\to 0$ limit for MB statistics is discussed for the fully-quantum hydrodynamic
equations (here represented by Eqs.\ \eqref{QHD}) and it is proven that such limit yields the
Madelung equations:
\begin{equation}
\label{MadEq}
\left\{ 
\begin{aligned}
 &\frac{\pt n}{\pt t} +  \ptx{i} (nu_i) = 0
 \\[4pt]
 &\frac{\pt u_i}{\pt t}  +  u_j \frac{\pt u_j}{\pt x_i}  + \frac{\pt V}{\pt x_i} 
   - \frac{\eps^2}{2} \ptx{i} \frac{\Delta \sqrt{n}}{\sqrt{n}} = 0.
\end{aligned}
\right.
\end{equation}
Equations \eqref{MadEq} were first derived by E.\ Madelung \cite{Madelung26}, 
and can be easily obtained from  Schr\"odinger equation by writing the wave function 
as $\psi = \sqrt{n} \,\e^{iS/\eps}$ and then putting $u = \nabla S$.
\subsubsection{BE case}
The discussion of the $T\to 0$ limit for Bose-Einstein statistics,  $\lam = -1$, is strongly dimension-dependent 
and we shall examine three cases, $d\geq3$, $d=2$ and $d=1$, separately.
\par
For $d\geq 3$, the condition \eqref{main_assum} is never satisfied in the limit of vanishing temperature
(physically speaking, all particles will be in the condensate phase)
and then such limit is a nonsense in our framework, because we are only considering
a completely non-condensate fluid (see Remark \ref{rem_condens}).
As it is well known, the mathematical description of the bose-Einstein condensate 
should be given in terms of a  nonlinear Schr\"odinger equation \cite{DalfovoEtAl99}.
\par
For $d=2$, the condensation does not occur and we can let $T$ go to 0 in the semiclassical
equations \eqref{semi_hydro} and \eqref{semi_diffu}.
As usual, we have to examine the asymptotic behavior of the functions $\phi_s^0(n)$.
This requires the inversion of $\phi_\frac{d}{2}(z)$ which, for $\lam = -1$ and $d = 2$,
is given by \cite{Lewin81}
$$
  \phi_1(z) = \Li_1(\e^z) = - \log(1 - \e^z).
$$
Recalling that $n_2 = 2\pi T$, we have
$$
   \phi_1^{-1}\left(\frac{n}{2\pi T}\right) = \log \left(1 - \e^{-\frac{n}{2\pi T}} \right)
   \sim  - \e^{-\frac{n}{2\pi T}}, 
   \quad
   \text{as $T \to 0$}.
$$
Now, since the relevant values of $s$ for the present case are $s = 0$, $s = -1$ and $s = -2$  
(i.e.\ $\frac{d}{2} -1$, $\frac{d}{2} -2$ and $\frac{d}{2} -3$), we can use \eqref{phi_prop2.2}
and conclude that
\begin{equation}
\label{phi0limBE2}
  \phi_s^0(n) \sim \Gamma(1-s)\,\e^\frac{(1-s)n}{2\pi T},
     \quad
   \text{as $T \to 0$}.
\end{equation}
By using \eqref{phi0limBE2} in Eqs.\ \eqref{semi_hydro} and \eqref{semi_diffu} it is readily seen that 
in order to obtain a finite limit we have to assume $R = 0$ and to rescale the density as
\begin{equation}
\label{resdens}
  \tilde n = \frac{n}{2\pi T},
\end{equation}
Then, it is not difficult to prove the following.
\begin{proposition}
Let $\lam = -1$ and $d = 2$, and assume $R=0$. 
Then, in the limit $T \to 0$, from the semiclassical hydrodynamic and diffusive
equations \eqref{semi_hydro} and \eqref{semi_diffu}
we formally obtain, respectively, the equations
\begin{equation}
\label{T0_BE_hydro2}
\left\{ 
\begin{aligned}
 &\frac{\pt \tilde n}{\pt t} +  \ptx{i} (\tilde n u_i) = 0
 \\[4pt]
 &\frac{\pt u_i}{\pt t}  +  u_j \frac{\pt u_j}{\pt x_i}   + \frac{\pt V}{\pt x_i} 
 - \frac{\eps^2}{12} \ptx{i} \Delta \tilde n = 0,
\end{aligned}
\right.
\end{equation}
and
\begin{equation}
\label{T0_BE_diffu2}
  \frac{\pt \tilde n}{\pt t} =  
  \ptx{i} \left( \tilde  n\frac{\pt V}{\pt x_i} - \frac{\eps^2}{12} \tilde  n \ptx{i} \Delta \tilde n \right)
\end{equation}
for the rescaled density \eqref{resdens}.
\end{proposition}
Note that in this case 
we have found a ``degenerate'' form of the limit Bohm potential, which reduces to a 
Laplacian.
\par
Let us finally examine the case $d=1$, which does not admit condensation as well.
From \eqref{phi_prop2.2} we have
$$
  \phi_\frac{1}{2}(z) \sim \frac{\sqrt{\pi}}{\sqrt{-z}},
  \quad
   \text{as $z \to 0^-$},
$$
and then, recalling that $n_1 = \sqrt{2\pi T}$, we can write
$$
  \phi_\frac{1}{2}^{-1}\left(\frac{n}{\sqrt{2\pi T}}\right) \sim - \frac{\pi}{n^2},
  \quad
   \text{as $T \to 0$}.
$$
We obtain, therefore, 
\begin{equation}
\label{phi0limBE1}
  \phi_s^0(n) \sim \Gamma(1-s)\left( \frac{n^2}{2 \pi^2 T} \right)^{1-s}, 
  \quad
  \text{as $T \to 0$},
\end{equation}
where now the relevant values are $s = -\frac{1}{2}$,  $s = -\frac{3}{2}$ and $s = -\frac{5}{2}$. 
By using \eqref{phi0limBE1}, and recalling that $R=0$ in the present one-dimensional case,
we see that the limit behavior of Eqs.\ \eqref{semi_hydro} and \eqref{semi_diffu} 
is non-singular, and the following is readily proven. 
\begin{proposition}
Let $\lam = -1$ and $d = 1$. 
Then, in the limit $T \to 0$, the semiclassical hydrodynamic and diffusive
equations \eqref{semi_hydro} and \eqref{semi_diffu} take, respectively, the form
\begin{equation}
\label{T0_BE_hydro1}
\left\{ 
\begin{aligned}
 &\frac{\pt n}{\pt t} +  \ptx{} (n u) = 0
 \\[4pt]
 &\frac{\pt u}{\pt t} +  u \frac{\pt u}{\pt x} +  \frac{\pt V}{\pt x} 
 - \frac{\eps^2 }{2} \ptx{} \left( \frac{1}{\sqrt{n}} \frac{\pt^2 \sqrt{n}}{\pt x^2} \right) = 0,
\end{aligned}
\right.
\end{equation}
and
\begin{equation}
\label{T0_BE_diffu1}
  \frac{\pt n}{\pt t} =  
  \ptx{} \left( n\frac{\pt V}{\pt x} 
  - \frac{\eps^2}{2} n \ptx{} \left( \frac{1}{\sqrt{n}} \frac{\pt^2 \sqrt{n}}{\pt x^2} \right) \right).
\end{equation}
\end{proposition}
Note that \eqref{T0_BE_hydro1} are the standard one-dimensional Madelung equations.
\section*{Acknowledgements}
This work was partially supported by the Italian Ministry of University,
PRIN  ``Mathematical problems of kinetic theories and applications'', prot.\ 2009NAPTJF\_003.
\appendix
\section{Moments of Fermi and Bose distributions and related integrals}
\label{polylog}
We recall that the {\em polylogarithm of order $s$}, with $s \in \mR$, is defined in the complex unit disc 
by the power series
$$
  \Li_s(z) = \sum_{k=1}^\infty \frac{z^k}{k^s}, \qquad \abs{z} < 1, 
$$
and can be analytically continued to  a larger domain (depending on $s$). 
To our purposes it will be enough to know that $\Li_s(z)$ is always well defined, real-valued
and regular for $z \in (-\infty, 1)$, and 
\begin{equation}
\label{RZ}
  \lim_{z \to 1^-} \Li_s(z) = 
  \left\{ 
\begin{aligned}
&\zeta(s),&   &\text{if $s>1$},
\\
 &+\infty,&  &\text{if $s \leq 1$,}
\end{aligned}
\right.
\end{equation}
where $\zeta$ is the Riemann zeta function.
The polylogarithms are strictly connected with the moments of FD and BE distributions.
\begin{definition}
\label{DefiA1}
For $\lam \e^z > -1$, $\lam \not= 0$, and $s \in \mR$ we define
\begin{equation}
\label{phiDef}
  \phi_s(z) = - \frac{1}{\lam} \Li_s (-\lam \e^z),
\end{equation}
where $\Li_s$ is the polylogarithm of order $s$.
From known identities \cite{Lewin81} we have that, for $s>0$, the above definition is equivalent to 
\begin{equation}
\label{phint}
  \phi_s(z) = \frac{1}{\Gamma(s)} \int_0^\infty \frac{t^{s-1}}{e^{t-z} + \lam}\,dt
   \qquad \quad
   \text{($s > 0$)}
\end{equation}
(known as {\em Fermi integral}).
\end{definition}
Note that $\phi_s(z)$ is defined for all $z\in\mR$ if $\lam > 0$, and for $z < -\log\abs{\lam}$ if $\lam < 0$.
In particular, in the FD case, $\lambda = 1$, $\phi_s(z)$ is defined on the whole real line while in the BE case,
$\lam = -1$, only for $z<0$.
\par
The following properties of the functions $\phi_s$ can be easily deduced from the 
properties of polylogarithms (see e.g.\  Refs.\ \cite{Lewin81,Wood92}):
\begin{subequations}
\begin{align}
\label{phi_prop1}
 &\lim_{\lam \to 0} \phi_s(z) = \e^z,&  &\text{for $ z \in \mR$ and $s \in \mR$;}
\\[4pt]
\label{phi_prop1.2}
 &\phi_s(z) \sim \e^z,&  &\text{for $z \to -\infty$, and $s \in \mR$;}
\\[4pt]
\label{phi_prop2}
 &\phi_s(z) \sim \frac{z^s}{\Gamma(s+1)},&  &\text{for $z \to +\infty$, $\lam = 1$ and $s \not= -1, -2, \ldots$;}
\\[4pt]
\label{phi_prop2.2}
 &\phi_s(z) \sim \Gamma(1-s)(-z)^{s-1},&  &\text{for $z \to 0^-$, $\lam = -1$ and $s < 1$;}
 \\[4pt]
\label{phi_prop3}
 & \frac{d}{dz}\phi_s(z) = \phi_{s-1}(z),&  &\text{for $\lam \e^z > -1$ and $s \in \mR$.}
\end{align}
\end{subequations}
Here, ``$f(x) \sim g(x)$ for $x \to y$'' means $\lim_{x \to y} f(x)/g(x) = 1$.
\par
\medskip
Starting from the identity \eqref{phint}, we shall now compute explicit expressions, in terms of the functions
$\phi_s$, of all the kinds of integrals that have been encountered throughout this paper.
\begin{lemma}
\label{lemma_moments}
For  $\lam \e^z > -1$, $k = 1, 2, 3, \ldots$  and $s>0$, let us consider the integrals
\begin{equation}
\label{Idef}
   I_k^s(z) = \frac{1}{\Gamma(s)} \int_0^\infty \frac{t^{s-1}}{\big(\e^{t-z} + \lam\big)^k}\,dt.
\end{equation}
Then, $I_1^s$ is given by Eq.\ \eqref{phint} and higher values of $k$ are recursively 
obtained by
\begin{equation}
\label{Ik}
   I_{k+1}^s(z) = \frac{1}{\lambda} \left(  I_k^s(z) - \frac{1}{k} \frac{dI_k^s}{dz}(z)\right).
\end{equation}
In particular (omitting the argument $z$),
\begin{equation}
\label{I1-4}
\begin{aligned}
&I_1^s = \phi_s
\\
&I_2^s = \lam^{-1} \left(\phi_s - \phi_{s-1}\right)
\\
&I_3^s =\lam^{-2} 
   \left(\phi_s - \textstyle{\frac{3}{2}}\phi_{s-1} + \textstyle{\frac{1}{2}}\phi_{s-2}\right)
\\
&I_4^s = \lam^{-3}
     \left(\phi_s - \textstyle{\frac{11}{6}}\phi_{s-1} + \phi_{s-2} - \textstyle{\frac{1}{6}}\phi_{s-3} \right).
\end{aligned}
\end{equation}
\end{lemma}
\begin{proof}
The recursive formula \eqref{Ik} follows immediately from a formal derivation of $I_k^s(z)$ with respect to $z$; 
Eq.\ \eqref{I1-4} follows from \eqref{phint} and \eqref{Ik} by using the property \eqref{phi_prop3}.
\end{proof}
Equation \eqref{I1-4} suggests that we can look for an expression for $I_k^s$ of this kind: 
\begin{equation}
\label{ICL}
   I_k^s = \frac{1}{\lambda^{k-1}} \sum_{j=0}^{k-1} c^k_j \phi_{s-j},
  \qquad
  k \geq 1,
\end{equation}
where $c^k_j$ are numerical coefficients (independent on $s$) to be determined.
Since
\begin{equation}
\label{DIs}
   \frac{d I_k^s}{dz}(z) = I_k^{s-1}(z)
\end{equation}
(as it is apparent from Lemma \ref{lemma_moments}),
from the recursive relation \eqref{Ik} we can write the equivalent relation
\begin{equation}
\label{newIk}
   I_{k+1}^s = \frac{1}{\lambda} \Big(  I_k^s - \frac{1}{k} I_k^{s-1} \Big).
\end{equation}
Inserting \eqref{ICL} into \eqref{newIk} yields
$$
  \sum_{j=0}^{k} c^{k+1}_j \phi_{s-j} = 
  \sum_{j=0}^{k-1} c^k_j \phi_{s-j} - \frac{1}{k}\sum_{j=1}^{k} c^k_{j-1} \phi_{s-j},
   \qquad
  k \geq 2.
$$
Equating the coefficients of $\phi_s$ ($j=0$) we obtain $c^{k+1}_0 = c^k_0$, 
and then (since $c^1_0 = 1$, as follows form \eqref{ICL} with $k=1$)
\begin{subequations}
\label{ccoeff}
\begin{equation}
\label{ccoeffA}
   c^k_0 = 1, \qquad k \geq 1;
\end{equation}
equating the coefficients of $\phi_{s-k}$ ($j = k$) we obtain $c^{k+1}_k = - \frac{1}{k} c^k_{k-1}$, and then
\begin{equation}
\label{ccoeffB}
   c^k_{k-1} = \frac{(-1)^k}{(k-1)!}, \qquad k \geq 1;
\end{equation}
finally, equating the coefficients of $\phi_{s-j}$, with $1\leq j \leq k-1$, we obtain
\begin{equation}
\label{ccoeffC}
   c^{k+1}_j = c^k_j - \frac{1}{k} c^k_{j-1}, \qquad k \geq 1,\quad  1\leq j \leq k-1.
\end{equation}
\end{subequations}
By using the recursive relations \eqref{ccoeffA}--\eqref{ccoeffC} one can easily  generate
all the coefficients of the expansion \eqref{ICL}.
\begin{proposition}
\label{prop_moments}
Let $I_k^s(z)$ be given as in the previous lemma.
Then, 
\begin{subequations}
\label{FDmom}
\begin{align}
\label{FDmomA}
& \int_{\mR^d} \frac{1}{\big(\e^{\frac{\abs{p}^2}{2T} - z} + \lam\big)^k}\,dp 
   = n_d I^{\frac{d}{2}}_k(z),
\\
\label{FDmomB}
& \int_{\mR^d} \frac{p_i p_j}{\big(\e^{\frac{\abs{p}^2}{2T} - z} + \lam\big)^k}\,dp 
    = \delta_{ij} n_d T I^{\frac{d}{2}+1}_k(z),
\\
\label{FDmomC}
&  \int_{\mR^d} \frac{p_i p_j p_k p_l}{\big(\e^{\frac{\abs{p}^2}{2T} - z} + \lam\big)^k}\,dp 
   = \left(\delta_{ij}\delta_{kl} + \delta_{ik}\delta_{jl}  + \delta_{il}\delta_{jk}  \right)  n_d T^2 I^{\frac{d}{2}+2}_k(z),
\end{align}
\end{subequations}
where, as usual, $n_d := (2\pi T)^\frac{d}{2}$.
\end{proposition}
\begin{proof}
By using polar coordinates it is easily shown that
\begin{equation}
\label{FDmomAUX}
   \int_{\mR^d} \frac{\abs{p}^{2m}}{\big(\e^{\frac{\abs{p}^2}{2T} - z} + \lam\big)^k}\,dp
    =  \frac{ n_d\,  (2T)^m \Gamma(\frac{d}{2}+m)}{\Gamma(\frac{d}{2})}\, I_k^{\frac{d}{2}+m}(z).
\end{equation}
This formula  immediately yields  Eq.\ \eqref{FDmomA} and, by obvious symmetry considerations, 
Eq.\ \eqref{FDmomB}.
The derivation of Eq.\ \eqref{FDmomC} requires more explanations.
We first show that
\begin{equation}
\label{case1}
  J_d(z) :=  \int_{\mR^d} \frac{p_i^4}{\big(\e^{\frac{\abs{p}^2}{2T} - z} + \lam\big)^k}\,dp
  = 3 n_d T^2 I_k^{\frac{d}{2}+2}(z)
\end{equation}
(which is of course independent on $i$).
We proceed by double induction on $d$. 
The cases $d=1,2$ can be easily verified by direct computations. 
Then, assuming \eqref{case1} to be valid for $d$, for $d+2$ we can write
$$
  J_{d+2}(z) =  
  \int_{\mR^2} J_d\left(z-\frac{\abs{q}^2}{2T} \right) \,dq
  = 3 n_d 2\pi T^2 \int_0^\infty I_k^{\frac{d}{2}+2} \left(z - \frac{\rho^2}{2T}\right) \rho \,d\rho.
$$ 
From Eq.\ \eqref{DIs} we obtain therefore
$$
 J_{d+2}(z) = 3 n_d 2\pi T^3 I_k^{\frac{d}{2}+3} \! \left(z - \frac{\rho^2}{2T}\right) \Big|^0_{+\infty}
 = 3 n_{d+2} T^2 I_k^{\frac{d+2}{2}+2}(z),
$$
which proves \eqref{case1} by induction.
On the other hand, Eq.\ \eqref{FDmomAUX} yields
$$
 \int_{\mR^d} \frac{\abs{p}^4}{\big(\e^{\frac{\abs{p}^2}{2T} - z} + \lam\big)^k}\,dp 
 = d(d+2)n_d T^2 I_k^{\frac{d}{2}+m}(z)
$$
and then, using 
$\abs{p}^4 = \left( \sum_i p_i^2 \right)^2
= \sum_i p_i^4 + \sum_{i\not=j} p_i^2p_j^2$ (and symmetry considerations), 
we obtain
\begin{equation}
\label{case2}
   \int_{\mR^d} \frac{p_i^2 p_j^2}{\big(\e^{\frac{\abs{p}^2}{2T} - z} + \lam\big)^k}\,dp
   = n_d T^2 I_k^{\frac{d}{2}+2}(z).
\end{equation}
for $i \not= j$.
The two cases, \eqref{case1} and \eqref{case2}, are summarized by \eqref{FDmomC}.
\end{proof}
\section{Postponed proofs}
\label{AppB}
\subsection{Proof of Proposition \ref{pr_Qclosure}}
\label{proof1}
According to the short notations adopted from Subsec.\ \ref{S3.1} on, 
let us denote $\cG_{A,B}$ by $\cG$ and $\frac{1}{T}h_{A,B}$ by $h$ (definition \eqref{newh}). 
Recalling, moreover, the formalism introduced in Sec.\ \ref{Sec_WBGK}, we can write
$$
  p_j \frac{\pt \cG}{\pt x_j}  = \frac{i}{\eps}\left\{\frac{1}{2}\abs{p}^2, \cG \right\}_\# = 
 \frac{i}{\eps} \left\{Th + p_j B_j - \frac{1}{2}\abs{B}^2 + A, \cG \right\}_\#
$$  $$
 =  \frac{i}{\eps} \left\{p_j B_j, \cG \right\}_\#
     +  \frac{i}{\eps} \left\{A - \frac{1}{2}\abs{B}^2, \cG \right\}_\#
  =  \frac{i}{\eps} \left\{p_j B_j, \cG \right\}_\#
    + \Theta_\eps\left[A - \frac{1}{2}\abs{B}^2\right] \cG,
$$
where we used the fact that $\Op_\eps(\cG)$ is, by definition \eqref{MEPg}, a function of 
$\Op_\eps\left(h \right)$ and then $\left\{h,\cG\right\}_\# = 0$, because it is the inverse 
Weyl quantization of a vanishing commutator. 
Since, from a direct computation, 
$$
  \frac{i}{\eps} \left\{p_j B_j, \cG \right\}_\# = 
  - p_j \frac{\pt B_j}{\pt x_k}  \frac{\pt \cG}{\pt p_k}  + B_j \frac{\pt \cG}{\pt x_j},
$$
then from the previous identity we have
$$
 \ptx{j} \bk{p_i p_j \cG} =  
 - \frac{\pt B_j}{\pt x_k}\bk{p_i p_j \frac{\pt \cG}{\pt p_k}}
 + B_j \ptx{j} \bk{p_i \cG} 
 + \bk{p_i\Theta_\eps\left[A - \frac{1}{2}\abs{B}^2\right] \cG }
$$ $$
 = \frac{\pt B_j}{\pt x_i} \bk{p_j \cG}
 + \frac{\pt B_j}{\pt x_j} \bk{p_i \cG}
 + B_j \ptx{j} \bk{p_i \cG} 
 + \bk{\cG} \ptx{i} \left( A - \frac{1}{2}\abs{B}^2 \right)
$$ $$
 = \frac{\pt B_j}{\pt x_i} J_j + \frac{\pt B_j}{\pt x_j} J_i
 + B_j  \frac{\pt J_i}{\pt x_j} 
 + n \left( \frac{\pt A}{\pt x_i}  - B_j \frac{\pt B_j}{\pt x_i} \right)
$$
(where Eq.\ \eqref{momTheta} was used), which yields Eq.\ \eqref{Qclosure}.
\hfill $\Box$
\subsection{Proof of Lemma \ref{lemmaTaylor}}
\label{proof_lemmaTaylor}
Lemma \ref{lemmaTaylor} follows from elementary manipulations of formal Taylor expansions.
In order to shorten the notations, let us introduce the $(d+1)$-dimensional vectors
$$
\begin{aligned}
 &m := (n,J_1,\ldots,J_d),
\\
 &\mu = (A,B_1,\cdots,B_d),
\\
 &f = f(\mu) = \left(\bk{\cG}, \bk{p_1 \cG},\ldots,\bk{p_d \cG}  \right),
\\
  &f^{(k)} = f^{(k)}(\mu) = \left(\bk{\cG_k}, \bk{p_1 \cG_k},\ldots,\bk{p_d \cG_k}  \right).
\end{aligned}
$$ 
The constraint system  \eqref{constr}, in these notations, reads as follows:
\begin{equation}
\label{compconstr}
    f(\mu) = m.
\end{equation}
Note that $f$ has a double dependence on $\eps$: one is direct (which leads to
the expansion \eqref{semiG}, i.e.\ to the terms $f^{(k)}$), 
and the other is through the Lagrange multipliers, which are expanded
as  $\mu = \mu^{(0)} + \eps\mu^{(1)} + \eps^2\mu^{(2)} + \cdots$. 
Then, we regard $f$ as $f(\eps,\mu(\eps))$, whose Taylor expansion at $\eps=0$ can be written
in this way:
\begin{multline*}
   f_i(\eps,\mu(\eps)) = f_i^{(0)}(\mu^{(0)}) + \eps\,\frac{\pt f_i^{(0)}}{\pt \mu_j}(\mu^{(0)}) \mu^{(1)}_j
\\
   + \eps^2 \left[\frac{1}{2}\frac{\pt^2 f_i^{(0)}}{\pt \mu_j  \pt\mu_k}(\mu^{(0)})  \mu^{(1)}_j \mu^{(1)}_k 
   + \frac{\pt f_i^{(0)}}{\pt \mu_j}(\mu^{(0)}) \mu^{(2)}_j
   + f^{(2)}_i(\mu^{(0)})\right]  + \cdots.
\end{multline*}
(where we took into account that $f^{(1)} = 0$, from \eqref{gexpB}, and, for the sake brevity,
the third-order term was not shown).
Since the moments $m$ do not depend on $\eps$, the constraint equation  \eqref{compconstr}
is expanded as follows:
$$
\begin{aligned}
   &f_i^{(0)}(\mu^{(0)})  = m_i
\\
  &\frac{\pt f_i^{(0)}}{\pt \mu_j}(\mu^{(0)}) \mu^{(1)}_j = 0
\\
 &\frac{1}{2}\frac{\pt^2 f_i^{(0)}}{\pt \mu_j  \pt\mu_k}(\mu^{(0)})  \mu^{(1)}_j \mu^{(1)}_k 
   + \frac{\pt f_i^{(0)}}{\pt \mu_j}(\mu^{(0)}) \mu^{(2)}_j
   + f^{(2)}_i(\mu^{(0)}) = 0
\\
 &\cdots
\end{aligned}
$$
The first equation is Eq.\ \eqref{expmuA}, the second one implies $\mu^{(1)} =0$
(i.e.\ $A^{(1)} = B^{(1)} = 0$), the third one is Eq.\  \eqref{expmuB} and, finally,
the fourth one (which has been omitted for brevity) implies $\mu^{(3)} =0$
(i.e.\ $A^{(3)} = B^{(3)} = 0$).
\hfill $\Box$
\subsection{Proof of Lemma \ref{lemmaG2}}
\label{proof_lemmaG2}
By using Eq.\ \eqref{gexpC} with $n=1$ we obtain
\begin{equation}
\label{g2short}
 \cG_2 = - \frac{\cG_0\, \Exp_2(h) + \cG_0\,\#_2\,\e^h}{\e^h + \lam}.
\end{equation}
The first term contains $\Exp_2(h)$, whose explicit expression can be taken from
Eq.\ (5.14) of Ref.\ \cite{DMR05} and reads as follows:
\begin{equation}
\label{Exp2}
 \Exp_2(h) = -\frac{\e^h}{8} \left( X_{ij} P_{ij} - S_{ij}S_{ji} + \textstyle{\frac{1}{3}} X_{ij}P_iP_j - 
 \textstyle{\frac{2}{3}} S_{ij}P_iX_j  + \textstyle{\frac{1}{3}} P_{ij}X_iX_j \right),
\end{equation}
where we introduced the following notations:
$$
\begin{aligned}
&X_i := \frac{\pt h}{\pt x_i} = -\frac{1}{T} \left[ (p-B)_k \frac{\pt B_k}{\pt x_i} + \frac{\pt A}{\pt x_i}\right]
\\
&P_i := \frac{\pt h}{\pt p_i} = \frac{1}{T}(p-B)_i
\\
&X_{ij} := \frac{\pt^2 h}{\pt x_i \pt x_j}
   = \frac{1}{T}\left[\frac{\pt B_k}{\pt x_i} \frac{\pt B_k}{\pt x_j} - (p-B)_k \frac{\pt^2 B_k}{\pt x_i \pt x_j}  
     - \frac{\pt^2 A}{\pt x_i \pt x_j}\right]
\\
&P_{ij} := \frac{\pt^2 h}{\pt p_i \pt p_j} = \frac{1}{T}\delta_{ij}
\\
&S_{ij} := \frac{\pt^2 h}{\pt x_i \pt p_j} = -\frac{1}{T}  \frac{\pt B_j}{\pt x_i}\,.
\end{aligned}
$$
(we recall that $h = \frac{1}{T} h_{A,B}$ is given by \eqref{newh}).
The other term, $\cG_0\,\#_2\,\e^h$, is a second-order Moyal product, given by \eqref{MoyalExpansion}.
Using the notations just introduced, we obtain:
$$
\begin{aligned}
\cG_0\,\#_2\,\e^h = &-\frac{1}{4} \left(\frac{1}{2} \frac{\pt^2 \cG_0}{\pt x_i \pt x_j}  \frac{\pt^2 \e^h}{\pt p_i \pt p_j} 
- \frac{\pt^2 \cG_0}{\pt x_i \pt p_j}  \frac{\pt^2 \e^h}{\pt p_i \pt x_j} 
+ \frac{1}{2} \frac{\pt^2 \cG_0}{\pt p_i \pt p_j}  \frac{\pt^2 \e^h}{\pt x_i \pt x_j}  \right)
\\
 =  &-\frac{\e^h}{8}\big[ \left(2F_2X_iX_j - F_1(X_{ij}+X_iX_j) \right)\left(P_{ij} + P_iP_j\right)
 \\
  &- 2\left(2F_2X_iP_j - F_1(S_{ij} + X_iP_j) \right)\left(S_{ij} + P_iX_j\right) 
 \\
  &+ \left( 2F_2P_iP_j - F_1(P_{ij}+P_iP_j )\right)\left(X_{ij} + X_iX_j\right) \big],
\end{aligned}
$$
where, for $k \geq 0$, we define.
\begin{equation}
\label{Fdef}
  F_k = \frac{\e^{kh}}{(\e^h + \lam)^{k+1}}.
\end{equation}
Putting together the two terms we obtain the following expression for $\cG_2$:
\begin{equation}
\label{g2}
\begin{aligned}
 \cG_2 &=
 \frac{1}{8} \left(X_{ij} P_{ij} - S_{ij}S_{ji} \right) \left(F_1 - 2F_2 \right)
 \\
 & +   \frac{1}{8} \left(X_{ij}P_iP_j -2S_{ij}P_iX_j + P_{ij}X_iX_j \right) 
 \left(\textstyle{\frac{1}{3}F_1 - 2F_2 + 2F_3} \right).
\end{aligned}
\end{equation}
Note that in the Maxwell-Boltzmann case, $\lam = 0$, we have 
$F_k = \e^{-h}$ for all $k\geq0$ and, therefore, 
$$
  F_1 - 2F_2 = -\e^{-h}, \qquad 
  \textstyle{\frac{1}{3}}F_1 - 2F_2 + 2F_3 = \textstyle{\frac{1}{3}}\e^{-h}.
$$
Then, in this case, Eq.\  \eqref{g2} reduces (as it should) to  $\Exp_2(-h)$, which is given 
by \eqref{Exp2}  with the suitable changes of sign.
\par
Now, defining $q = p-B$, taking the moments $\bk{\cG_2}$ and  $\bk{q_i \cG_2}$ from Eq.\ \eqref{g2},
and taking account of vanishing integrals of odd functions, we get
\begin{equation}
\label{G2mom0}
\begin{aligned}
\bk{\cG_2} &=  \frac{1}{8T^2} \left[  \frac{\pt B_j}{\pt x_k}  
  \left( \frac{\pt B_j}{\pt x_k} - \frac{\pt B_k}{\pt x_j} \right) - \Delta A \right] \bk{F_1 - 2F_2}   
\\
  &+ \frac{1}{8T^3}  \left[  \frac{\pt B_j}{\pt x_k}  
  \left( \frac{\pt B_j}{\pt x_l} - 2\frac{\pt B_l}{\pt x_j} \right) -  \frac{\pt^2 A}{\pt x_k \pt x_l}
  + \frac{\pt B_k}{\pt x_j} \frac{\pt B_l}{\pt x_j}\right] 
 \bk{q_k q_l \left(\textstyle{\frac{1}{3}F_1 - 2F_2 + 2F_3} \right)}
\\
  &+ \frac{1}{8T^3} \abs{\nabla A}^2 \bk{ \textstyle{\frac{1}{3}F_1 - 2F_2 + 2F_3}}
\end{aligned}
\end{equation}
and
\begin{equation}
\label{G2mom1}
\begin{aligned}
\bk{q_i\cG_2} &=  
   - \frac{1}{8T^2} \Delta B_j  \bk{ q_i q_j \left(F_1 - 2F_2\right)}
\\
  &- \frac{1}{4T^3}  \frac{\pt A}{\pt x_k} \left( \frac{\pt B_k}{\pt x_j} - \frac{\pt B_j}{\pt x_k} \right)
   \bk{q_i q_j\left(\textstyle{\frac{1}{3}F_1 - 2F_2 + 2F_3} \right)}
\\
  &- \frac{1}{8T^3} \frac{\pt^2 B_l}{\pt x_j \pt x_k}
  \bk{ q_i q_j q_k q_l\left(\textstyle{\frac{1}{3}F_1 - 2F_2 + 2F_3} \right)}.
\end{aligned}
\end{equation}
The moments of functions $F_k$ can be reduced to integrals of type $I_k^s$ (see Lemma 
\ref{lemma_moments}) by using
\begin{equation}
\label{FvsI}
F_k 
=  \frac{1}{\e^h + \lam}\left(1 - \frac{\lam}{\e^h + \lam} \right)^k 
= \sum_{j=0}^k {k \choose j} \frac{(-\lam)^j}{(\e^h + \lam)^{j+1}}. 
\end{equation}
Recalling that
$$
h = \frac{\abs{q}^2}{2T} - \frac{A}{T}, 
$$
from \eqref{FvsI} and\eqref{FDmom} we obtain
$$
\begin{aligned}
& \bk{F_1 - 2F_2} = -n_d \phi_{\frac{d}{2}-2} \Big(\frac{A}{T}\Big) 
\\
& \bk{ q_i q_j \left(F_1 - 2F_2\right)} = - \delta_{ij} n_d T  \phi_{\frac{d}{2}-1} \Big(\frac{A}{T}\Big) 
\\  
& \bk{ \textstyle{\frac{1}{3}F_1 - 2F_2 + 2F_3}} = \frac{1}{3}n_d\phi_{\frac{d}{2}-3}  \Big(\frac{A}{T}\Big) 
\\
& \bk{q_i q_j\left(\textstyle{\frac{1}{3}F_1 - 2F_2 + 2F_3} \right)} 
    =  \frac{1}{3} \delta_{ij} n_d T  \phi_{\frac{d}{2}-2}  \Big(\frac{A}{T}\Big) 
\\
& \bk{ q_i q_j q_k q_l\left(\textstyle{\frac{1}{3}F_1 - 2F_2 + 2F_3} \right)} 
    =  \frac{1}{3}\left(\delta_{ij}\delta_{kl} + \delta_{ik}\delta_{jl}  + \delta_{il}\delta_{jk}  \right) 
        n_d T^2  \phi_{\frac{d}{2}-1} \Big(\frac{A}{T}\Big) .
\end{aligned}
$$
Then, by using $\bk{p_i\cG_2} =  B_i\bk{\cG_2} + \bk{q_i\cG_2}$ and the identity
$$
  \frac{1}{T} \frac{\pt A}{\pt x_k} \,\phi_{\frac{d}{2}-2} \Big(\frac{A}{T}\Big) 
  = \ptx{k}\, \phi_{\frac{d}{2}-1} \Big(\frac{A}{T}\Big) 
$$ 
(from property \eqref{phi_prop3}), Eqs.\ \eqref{momG2} are easily obtained from  the expressions 
\eqref{G2mom0} and \eqref{G2mom1}.
\par
\hfill $\Box$
\end{document}